\documentclass[sts]{imsart}

\RequirePackage{amsthm,amsmath,amsfonts,amssymb}
\RequirePackage[numbers,sort&compress]{natbib}
\RequirePackage[colorlinks,citecolor=blue,urlcolor=blue]{hyperref}
\RequirePackage{graphicx}
\RequirePackage{subcaption}

\startlocaldefs

\def\RR{\mathbb{R}}

\def\E{\mathbb{E}}
\def\FF{\mathcal{F}}
\def\GG{\mathcal{G}}
\def\NN{\mathcal{N}}
\def\P{\mathbb{P}}
\def\P{\mathbb{P}}

\def\argmax{\mbox{argmax}}

\theoremstyle{plain}

\newtheorem{theorem}{Theorem}[section]
\newtheorem{lemma}[theorem]{Lemma}

\newtheorem{proposition}[theorem]{Proposition}
\theoremstyle{definition}

\endlocaldefs

\begin{document}
\begin{frontmatter}
\title{A Note on Robustifying Empirical Bayes} 
\runtitle{Robustifying Empirical Bayes}

\begin{aug}
\author[A]{\fnms{Roger}~\snm{Koenker}\ead[label=e1]{r.koenker@ucl.ac.uk}}
\author[B]{\fnms{Jiaying}~\snm{Gu}\ead[label=e2]{jiaying.gu@utoronto.ca}}
\thankstext{t1}{}
\address[A]{Roger Koenker is Honorary Professor, 
	Department of Economics, University College London, London, WC1H OAX, U.K.
	\printead[presep={\ }]{e1}.}
\address[B]{Jiaying Gu is Associate Professor,
	Department of Economics, University of Toronto, Toronto, M5S 3G7, Canada
	\printead[presep={\ }]{e2}.}
\end{aug}

\begin{abstract}
Two strategies are explored for robustifying classical denoising  procedures for the
Gaussian sequence model.  First, the Hodges and Lehmann (1952)
restricted Bayes approach is used to reduce sensitivity to the specification
of the initial prior distribution.  Second, alternatives to the Gaussian
noise assumption are explored.  In both cases proposals of Huber (1964)
and Mallows (1978) play a crucial role.
\end{abstract}


\end{frontmatter}
	
\section{Introduction}

The Gaussian sequence model can be viewed as a compound decision problem with
observed $X_i \sim \NN (\theta_i, 1), \; i = 1, \dots , n$.  The objective is to estimate the 
$\theta \in \RR^n$ subject to quadratic loss.  We will denote the standard Gaussian
density and cumulative by $\varphi$ and $\Phi$ respectively.  Observations are assumed to be 
exchangeable, so their marginal density is given via de Finetti by, 
\[
f_G (x) = \int \varphi (x| \theta) dG(\theta),
\]
for some mixing distribution $G$.
If the $\theta$'s are viewed as fixed parameters, then $G$ may be regarded as their (unobserved) 
empirical distribution.  Were $G$ known the optimal (Bayes) decision rule is given by Tweedie's 
formula, \cite{efron.11}
\[
\hat \theta_i = \delta^B (x_i) = x_i + f_G'(x_i)/f_G (x_i).
\]

When $G$ is unknown various shrinkage procedures have been proposed, initiated by the
fundamental papers of \cite{stein56} and \cite{robbins56}.  The extensive literature on
Stein shrinkage has offered a rich assortment of practical frequentist and Bayesian procedures
for improving upon the naive maximum likelihood estimator, $\delta (x_i) = x_i$ in terms
of quadratic loss.  Among these procedures more recently the
nonparametric maximum likelihood estimator (NPMLE) of \cite{kw},
\[
\hat G = \argmax_{G \in \GG} \big \{ \sum_{i=1}^n \log (f_G (x_i) \big \},
\]
has been proposed as a plug-in estimator for $G$, \cite{jz}, \cite{km} and \cite{soloff}.
This $G$-modeling strategy -- in the terminology of \cite{e19} -- performs well in simulations, 
e.g. \cite{km}, \cite{gk16}, \cite{kg26},  relative to alternatives that attempt to 
estimate $f_G$ directly or
that make a priori assumptions about the form of $G$.  However, it is obviously subject
to the criticism that the Gaussian assumption on the likelihood is quite strong,
and priors are never terribly convincing.

In what follows we consider two basic strategies for robustifying empirical Bayes
procedures.  The first, following a proposal of \cite{hodgeslehmann}, seeks protection
from excessive confidence in our initial prior on $G$ by bounding pointwise risk
thereby offering a compromise between minimax and Bayes decision rules.  The second, 
acknowledges scepticism about the strictly Gaussian form of $\varphi$, the distribution
of the model noise.
We find that that modest modifications of
of our initial prior or the Gaussian noise assumption can significantly improve performance 
of empirical Bayes decision rules while sacrificing only modest performance in the 
event that the initial prior or the Gaussian noise assumptions are valid.

\section{Bayes Risk and Brown's Identity} 
Bayes risk in the Gaussian sequence model is,
	\[
	r(G, \delta) = \int R(\delta, \theta) dG(\theta),
	\]
with 
	\[
	R(\delta, \theta) =\mathbb{E}_\theta[(\delta(X) - \theta)^2] =  
	\int ( \delta(x)-\theta)^2 \varphi(x - \theta) dx. 
	\]
Plugging the optimal Bayes rule back into $r(G, \delta)$, we obtain
Brown's identity, \cite{Brown71}:
	\begin{align*}
	r(G, \delta^{B}) & =  
	\int (x-\theta + \frac{f_G'(x)}{f_G(x)})^2 \varphi(x - \theta) dx dG(\theta)\\
	& = 1 + 2 \mathbb{E}\Big [\Big (\frac{f_G'(X)}{f_G(X)}\Big )'\Big ] + 
	\mathbb{E}\Big [\Big (\frac{f_G'(X)}{f_G(X)}\Big )^2 \Big ]\\
	& = 1 +2 \mathbb{E}\Big [ \frac{f_G^{''}(X)}{f_G(x)} - \Big( \frac{f_G'(x)}{f_G(x)}\Big)^2 \Big ]
	 + \mathbb{E}\Big [\Big (\frac{f_G'(X)}{f_G(X)}\Big )^2 \Big ]\\
	& = 1- \mathbb{E}\Big[ \Big (\frac{f_G'(X)}{f_G(X)}\Big )^2 \Big ]\\
	& = 1 - I(\Phi * G).
	\end{align*}
The second equality follows from Stein's lemma, the fourth from the fact that 
$\int f_G^{''} (x) dx = 0$, 
and the last equality from the definition of Fisher
information for distributions with absolutely continuous densities.

It may seem curious that Bayes risk of the optimal empirical Bayes rule for the Gaussian
sequence model reduces to Fisher information for a scalar location parameter of the convolution
distribution $\Phi * G$.  We will exploit the latter connection in the next section to consider least favorable
alternative priors for an initial prior in which we lack complete confidence.  Such modified priors
offer some compromise between strictly Bayesian and minimax procedures.
Choosing contamination models by minimizing Fisher information 
is a classical strategy for choosing alternatives for the univariate Gaussian location 
and regression problems following \cite{huber64}.  The convolution form of the Fisher
information for Bayes risk leads us back to an alternative proposal of \cite{mallows78}
as well.

A natural objection to many empirical Bayes procedures is that they place unjustified reliance on an initial
prior.  While such procedures may still perform well with respect to ensemble risk, they may also fail spectacularly
for some subpopulations or individuals.  This concern underlies the limited translation
proposal of \cite{efronmorris}.  We will see that bounding minimax risk by modifying an initial
prior can often serve to soften the impact of these failings.

\section{Restricted Bayes Solutions}

In an effort to balance minimax and Bayes solutions, \cite{hodgeslehmann} proposed solving,
\[
\min_\delta r(G_0,\delta) \; s.t. \;  \max_\theta R(\delta, \theta) \leq 1 + t,
\]
for an initial prior $G_0$ and some $t > 0$. Since $\max_\theta R(X,\theta)=1$ corresponding to 
the worst pointwise risk among all decision rules, achieved by the MLE estimator $\delta(X) = X$, the 
Hodges and Lehmann modified decision rule is thereby constrained to do uniformly 
well over the entire parameter space with risk bounded by $1+t$, while minimizing 
the Bayes risk under prior $G_0$, hence called the restricted Bayes rule.  They show that this is equivalent
to solving,
\[
\min_\delta \; \max_{G \in \GG_\epsilon(G_0)} r(G, \delta), 
\]
with $\GG_\epsilon(G_0) = \{ G = (1 - \epsilon) G_0 + \epsilon H \}$ for some $\epsilon \in (0,1)$ depending
upon $t$, where $H$ is an arbitrary distribution for $\theta$. Due to the minimax theorem we can switch the 
minimization and the maximization, and the resulting optimal rule $\delta^*$ is the posterior mean 
of $\theta$ under the least favorable prior from the class $\mathcal{G}_{\epsilon}(G_0)$. \cite{berger85} comments that
``it is \emph{very} difficult to determine such $\delta^*$; furthermore, this 'optimal' $\delta^*$ is 
usually extremely messy and difficult to work with.''  On the contrary, with the aid of 
modern convex optimization techniques we find them quite tractable and elegant.  

\cite{bickel83} considers the case with $G_0$ having point mass one at zero.  Then, by the 
Brown identity, the Hodges and Lehmann problem is equivalent to solving,
\[
\max_{G \in \GG_\epsilon(\delta_0)} (1 - I (\Phi * G)) = \min_{G \in \GG_\epsilon(\delta_0)}  I (\Phi * G)
\]
This is the problem posed by \cite{mallows78} motivated by robustness considerations for
time-series problems with additive outliers.
Mallows conjectured that the least favorable $G$ would be discrete,
supported on the integers with mass declining exponentially.  
\cite{bickel83} reports a modified conjecture of Donoho that relaxes the spacing
of the Mallows mass points, but is otherwise similar.  Neither conjecture seems to be strictly 
correct, but numerical computations confirm the nearly exponential decay of the mass. 
\cite{bickel1983minimizing} provide a detailed discussion of the discrete nature of the Mallows
solutions based on the analyticity of the objective function.  See also \cite{johnstone94}.  

\subsection{Computing Mallows's Least Favorable Distribution} \label{sec: computation}

As noted by \cite{bickel1983minimizing} and \cite{marazzi} the \cite{mallows78} problem is convex.
\cite{marazzi} suggests a gridding strategy that imposes an exponentially declining mass condition.
This produces a remarkably accurate solution for an initial 
Gaussian prior employing generic optimization software. Using modern convex optimization software, 
we show that the problem can be efficiently solved numerically for any prior distribution. 
Our implementation embodied in the function \texttt{HodgesLehmann} in the our REBayes package 
for the R language employs the Mosek \cite{mosek} optimizer and provides a general interface for 
computing either the Huber or Mallows solutions for an arbitrary initial prior, $G_0$. 

We now describe our procedure for the simplest (Dirac) initial prior, $G_0 = \delta_0$.
Our objective is to solve 
\[
\underset{f \in \mathcal{K} }{\min} \int \frac{f'(x)^2}{f(x)} dx 
\]	
with 
\[
\mathcal{K}_\epsilon =\Big \{f = (1 - \epsilon) \int \varphi(x- \theta)d \delta_0(\theta)  + 
\epsilon \int \varphi(x- \theta) dH(\theta)\Big  \} 
\]
This can be solved, on a grid of $x$, $\{x_1< x_2 < \dots <x_M\}$ 
and a grid of $\theta$ as $\{ \theta_1 < \theta_2 < \dots < \theta_L\}$
as the rotated quadratic cone convex optimization problem:
\[
\min \sum_{i = 1}^M w_i 
\]
subject to 
\begin{align*}
u_i &= f_{i+1}-f_i, \\
v_i  &= \frac{1}{2}(	f_{i+1}+f_i),\\
u_i^2 &\leq 2v_i w_i\\
f_i &= (1 - \epsilon)  \varphi(x_i) + \epsilon \sum_{j=1}^L \varphi(x_i - \theta_j) h_j\\
h & \in \mathcal{S} \equiv \{h \in \mathbb{R}^L : \sum_{j=1}^L h_j = 1, \; h_j \geq 0, \; j = 1, \dots , L \} 
\end{align*}
Provided that the grids are sufficiently finely spaced interior point optimization in Mosek 
is capable of producing very accurate solutions very efficiently.

In Figure \ref{fig.Mallows} we illustrate a Mallows marginal density, $f^M$, its corresponding
mixing distribution, $P^M$, and plot the log mass of the discrete mass points of $P^M$ at their
respective locations.  At first glance, it seems that the mass points are approximately equally spaced and
have mass that declines exponentially.  However, on closer examination the spacing of
the mass points in the right tail are estimated to be:  1.96, 1.80, 1.70, 1.61, 1.52, 
1.39, 1.29, 1.37, 1.55, 1.70, 1.91, which seems sufficiently non-uniform to call the
uniform spacing conjecture into question.  \cite{djm} suggest an alternative computational
strategy for the Mallows problem using a parametric model that assumes equal spacing of the
mass points.  The grid for evaluation of $f$ is equally spaced
from -30 to 30 with 500 points of evaluation.  The grid for evaluation of $P^M$ is also equally
spaced from -20 to 20 with 4003 points of evaluation.  For purposes of illustration the mass
at $\theta = 0$ is taken to be 0.2.

\begin{figure*}
  \begin{center}
          \resizebox{ \textwidth}{!}{{\includegraphics{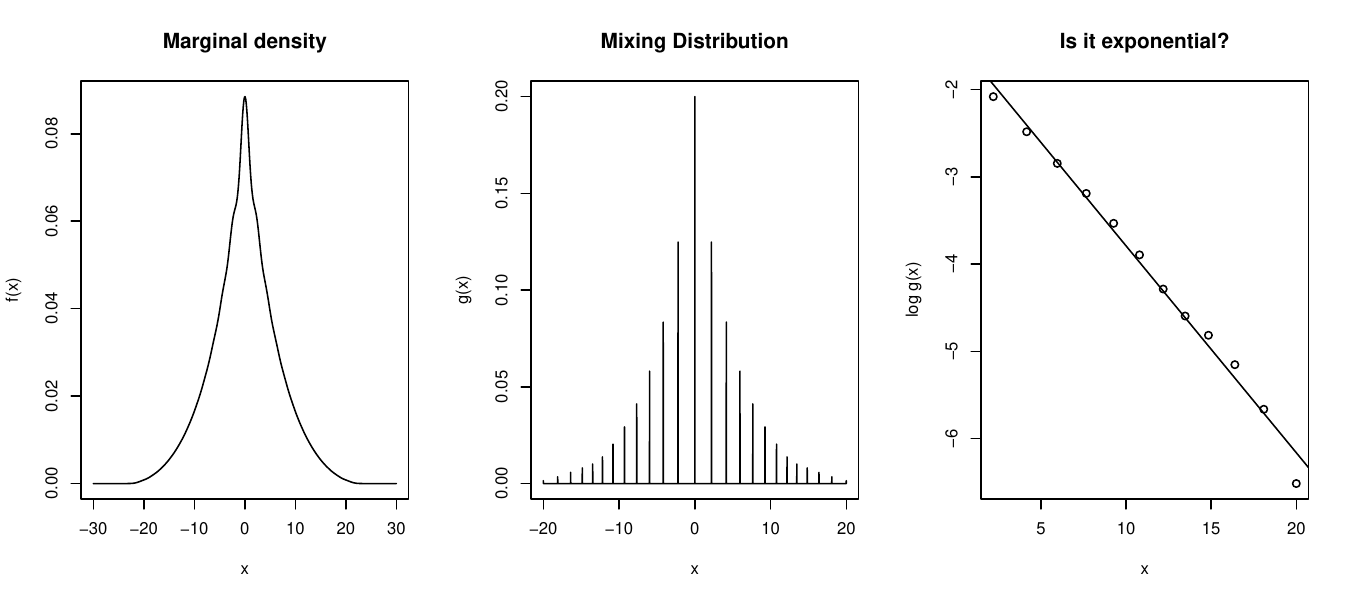}}}
  \end{center} 
  \caption{Mallows least favorable marginal density, probability mass function of the Mallows
  mixing distribution and log mass of the
  mixing distribution as a function of location indicating the approximate exponentiality
  of the mixing distribution as conjectured by Mallows. The prior $G_0$ is a Dirac at zero and $\epsilon = 0.8$. }
  \label{fig.Mallows}
\end{figure*}

\begin{figure}
  \begin{center}
          \resizebox{ .4\textwidth}{!}{{\includegraphics{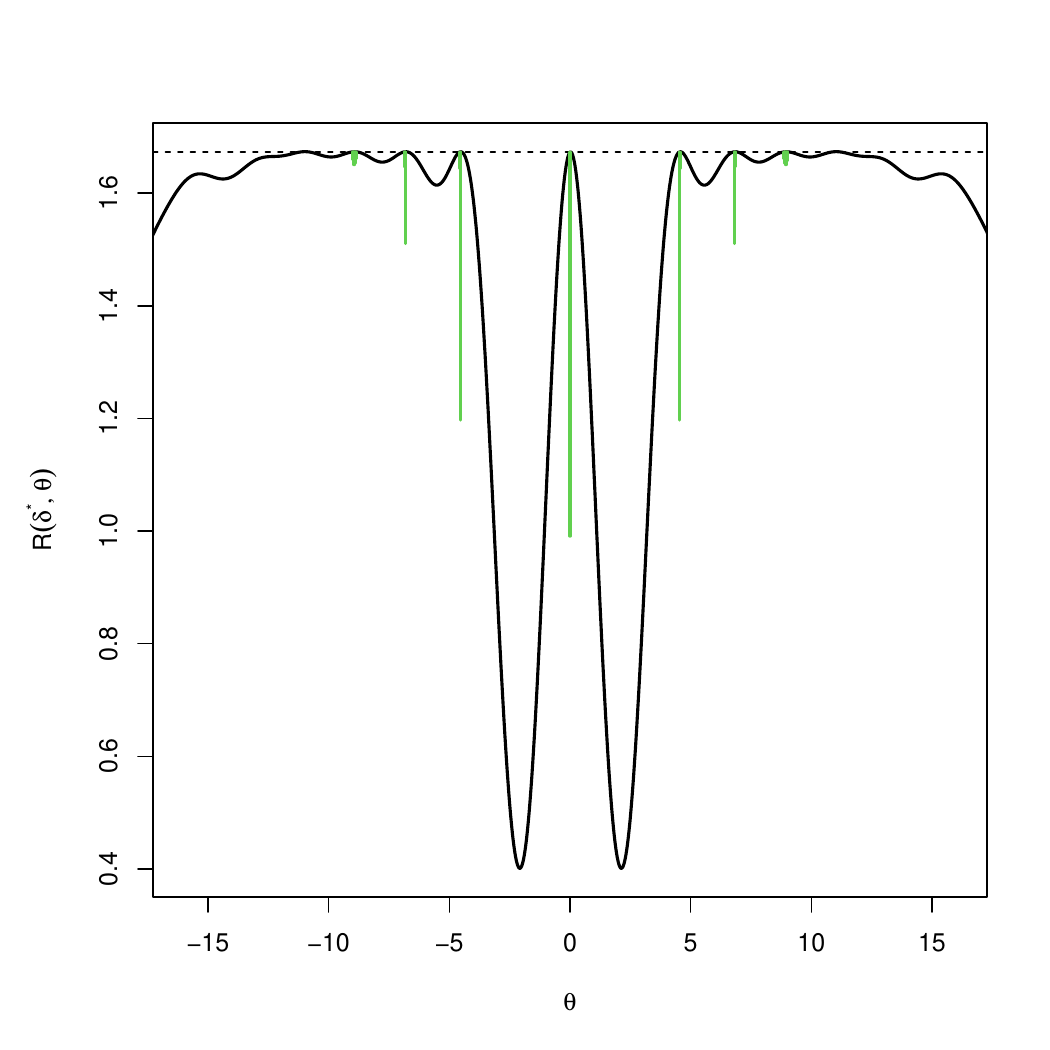}}}
  \end{center} 
  \caption{Mass points of the Mallows contamination distribution $H^*$ with 
  $G_0 = 0.5 \delta_{-2} + 0.5 \delta_2$ for $\epsilon = 0.2$.  The solid black
  curve depicts the pointwise risk function of the decision rule $\delta^*$, constructed using Mallows's least favorable prior $G^*$. The vertical green lines depict the location of the mass points of
  the solution $H^*$, while their length represents the amount of mass assigned to each.  In
  accordance with the Tukey ``hanging rootogram'' principle these lengths are rescaled as the square
  root of the respective masses.  The dashed horizontal line represents the bound on
  the pointwise risk imposed by the Hodges-Lehmann constraint in this case approximately
  1.67 induced by the choice of $\epsilon = 0.2$. The constraint is binding at mass points of $H^*$.}
  \label{fig.HLmass}
\end{figure}

In the previous example we have taken the initial prior, $G_0$ as Dirac, but 
there is no obstacle to starting from any other initial prior.  To provide some additional
intuition about the nature of the Mallows solution for general $G_0$, we illustrate in Figure \ref{fig.HLmass}
a plot of the pointwise risk function $R(\delta^* , \theta) := \mathbb{E}_\theta[(\delta^*(X)-\theta)^2]$ of the decision rule $\delta^*(\cdot)$, constructed as the posterior mean of $\theta$ using the least favorable Mallows prior
\[
G^* (\theta) = (1 - \epsilon ) G_0 (\theta) + \epsilon H^* (\theta)
\]
where $G_0 (\theta)$ is taken to be an equally weighted mixture of two point masses at -2 and 2.
The tangencies in this plot with the horizontal dotted line marking out $\sup_\theta R(\delta^*, \theta) = 1 + t$ coincide 
with the location of the mass points of the solution of $H^*$ indicated in the plot by the
vertical green lines.  The bound, $1 + t$, is the pointwise
risk bound chosen to constrain the Hodges-Lehmann restricted Bayes rule, which can be constructed using Mallows's least favorable prior $G^*$. 
Since the initial prior $G_0$ places all its mass on the two points $\{-2,2\}$ the Mallows modification
hedges this bet by placing a considerable mass at zero and exponentially declining mass at a few points
below -2 and above +2. This figure is strongly reminiscent of Figure 5.5 of \cite{lindsay} 
illustrating the location of mass points of the NPMLE. 

\subsection{Some Examples}

We now consider several special cases of the Hodges and Lehmann restricted Bayes approach.
In each case we consider not only the Mallows equivalent form of the Hodges-Lehmann modification, but 
also a Huber alternative that relaxes the Mallows objective of minimizing the Fisher information over convolutions by 
minimizing over the entire class of contamination distributions for $X$.  Taking the least favorable density and plug into the Tweedie formula gives rise the Huber procedure to estimate $\theta$ for each value of $x$. 

The Mallows rule has the obvious advantage that it yields a Bayes decision rule while
the corresponding Huber procedure does not.  This is particularly evident in the third
example of Casella and Strawderman where the unrestricted Bayes rule is minimax, so the Hodges-Lehman's 
restriction on point-wise risk is unbinding and the restricted Bayes rule coincides with the unrestricted, 
but the Huber procedure is inadmissible.  On the other hand there is something attractive about the Huber 
rules that it can be shown that they are necessarily monotone. \cite{donohoreeves} propose
an alternative based on the \cite{huber74} spline that minimizes Fisher information over a Kolmogorov
neighborhood of the marginal density of $X$ specified by a finite number of evaluations of its quantile function. 
They then apply the Tweedie formula with the resulting least favorable density. 
An implementation of this procedure is also included in the REBayes package with the function
\texttt{HuberSpline}, although we do not pursue it further in this paper.

    \vspace{5mm}
\begin{description}
\item[Dirac $G_0$]  This is the case considered by \cite{bickel83} and anticipated by \cite{mallows78}.
	Bickel first simplifies the problem of minimizing $I(\Phi * G)$ by considering the relaxed 
	\cite{huber64} problem of minimizing $I(F)$ over $\FF = \{ F = (1 - \epsilon) \Phi + \epsilon W\}$ where $W$ is an arbitrary distribution for $X$.
	The least favorable distribution $F$ has the well known score function,
	\[
	-f^\prime (x)/f(x) = 
	    \begin{cases} x & |x| \leq k\\ 
	    k \; \text{sign}(x) & |x| > k.
	    \end{cases}
	\]
	Applying Tweedie's formula gives rise to the hard thresholding Huber rule,
	\[
	\delta(x) = x + f^\prime (x)/f(x) = 
	    \begin{cases} 0 & |x| \leq k\\
	    x - k \; \text{sign}(x) & |x| > k
	    \end{cases} 
	\]
	We contrast this with the numerical solution of the corresponding 
	Mallows problem in Figure \ref{fig.dirac}.  The Mallows rule offers a soft thresholding
	alternative to the piecewise linear Huber rule that oscillates around the Huber rule
	in the tails.
    \begin{figure}
    \begin{center}
          \resizebox{.4\textwidth}{!}{{\includegraphics{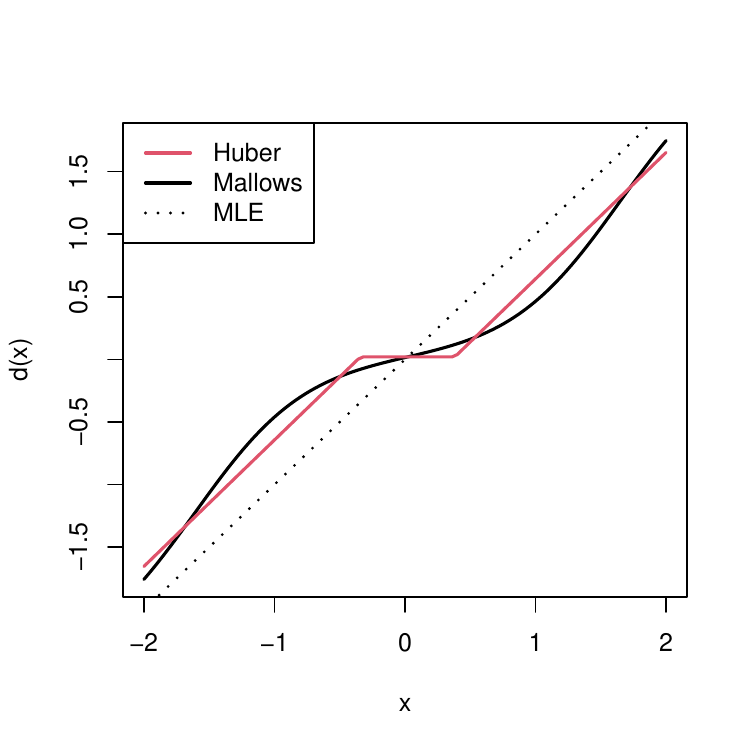}}}
    \end{center} 
    \caption{The figure contrasts the Huber and Mallows  restricted Hodges-Lehmann decision rules
    for an initial Dirac prior and $\epsilon = 0.6$. The Huber modification imposes hard thresholding 
    near zero while the Mallows rule is smoother oscillating around the Huber rule.}  
    \label{fig.dirac}
    \end{figure}
    \begin{figure}
    \begin{center}
          \resizebox{.4\textwidth}{!}{{\includegraphics{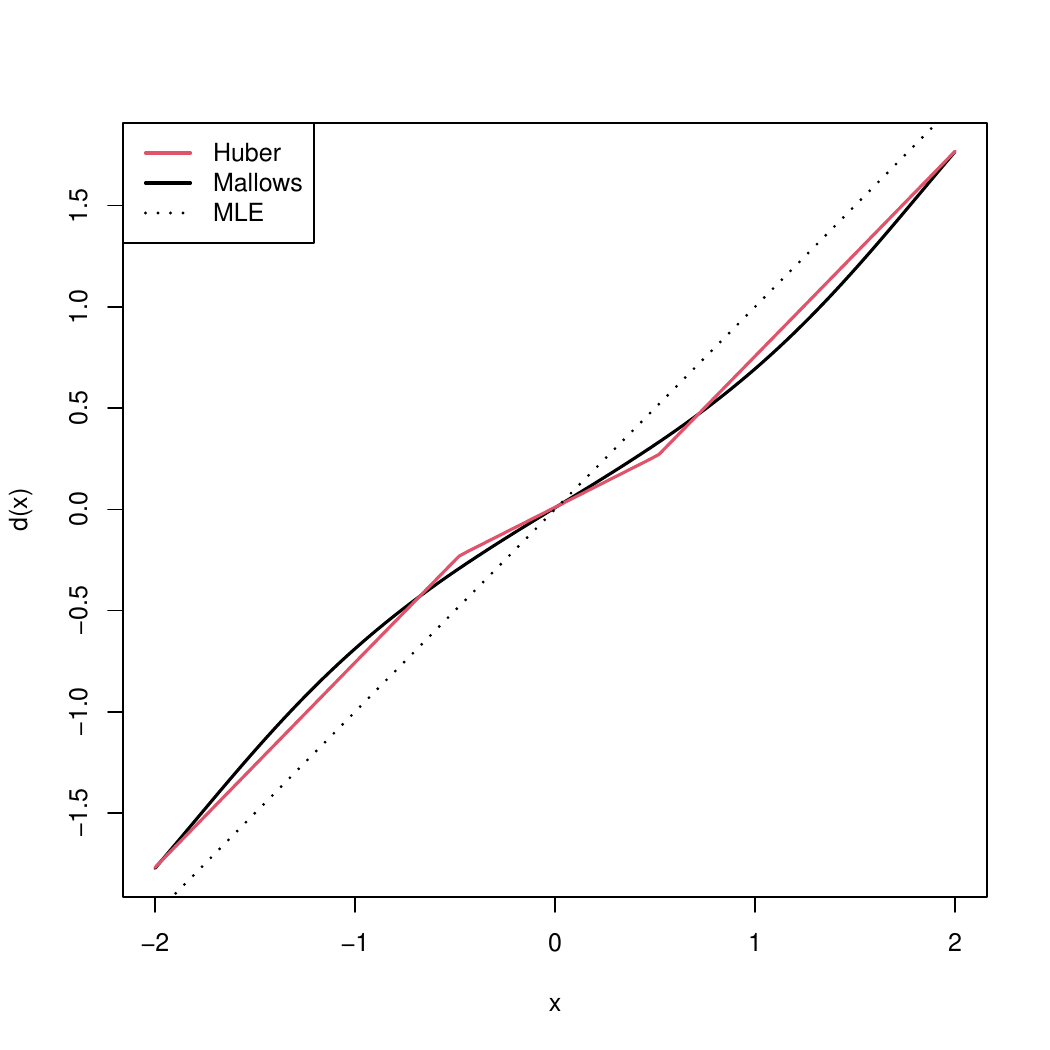}}}
    \end{center} 
    \caption{The figure contrasts the Huber and the Mallows rules for an initial Gaussian prior. The Huber rule yields the Efron-Morris piecewise linear limited translation rule
    while the Mallows rule is smoother near zero and oscillates around the Huber rule in
    the tails.  In both cases, $\epsilon = 0.6$.}\label{fig.gauss}
    \end{figure}
    \vspace{5mm}
\item[Gaussian $G_0$]  This is the case considered by \cite{efronmorris}.  The initial prior is
    Gaussian, $G_0 \sim \NN(0, A)$, so the score function of the least favorable Huber density, which minimizes Fisher information of distributions in the contamination class $\mathcal{F}_\epsilon = \{ F = (1-\epsilon) N(0,A+1) + \epsilon W\}$, is
    \[
    -f^\prime (x)/f(x) = 
    	\begin{cases} \frac{1}{A+1}x & |x| \leq k(A+1)\\
	k \; \text{sign}(x) & |x| > k(A+1)
	\end{cases}
    \]
    and Tweedie's formula yields the piecewise linear ``limited translation'' rule,
    \begin{align*}
    \delta(x) & = x + f^\prime (x)/f(x)\\ 
    & = \begin{cases} \frac{A}{A+1} x & |x| \leq k(A+1)\\
	x - k \; \text{sign} (x) & |x| > k(A+1),
	\end{cases} 
    \end{align*}
    proposed by Efron and Morris.
    Figure \ref{fig.gauss} contrasts the Huber and Mallows forms of the Hodges-Lehmann restricted Bayes rule 
    with the unrestricted Bayes rule. We stick to linear shrinkage for values of $x$ in the middle but refrain 
    from shrinking in the two tails.  Again, the Mallows rule smooths the Huber rule in the center 
    and oscillates around the Huber rule in the tails.  Remarkably, \cite{efronmorris} in
    their Appendix Figure B' already illustrates this behavior of the restricted Mallows rule. 
    See also Figure 1 in \cite{marazzi}.  Extension of this example to James-Stein forms for 
    $G_0$ is straightforward although computation of the associated risk is not as simple as shown by 
    \cite{efronmorrisii}.
\item[Two Point $G_0$ I]  This is the case considered by \cite{casella1981estimating}, the
    initial prior is $G_0 = 0.5 \delta_{-1} + 0.5 \delta_{1}$ and we restrict the domain
    of the parameter $\theta$ to the interval $[-1,1]$.  Casella and Strawderman show 
    that the Bayes rule,
    \[
    \delta^B (x) = \tanh (x)
    \]
    attains maximal pointwise risk of 
    \[
    R(\delta^B, 1) = \E_{X \sim \NN(1,1)} [ \tanh(X) - \theta)^2] \approx 0.45 < 1,
    \]
    at $\theta = \pm 1$.  Consequently, $\delta^B$ is minimax and $G_0$ is least favorable for
    all $t(\epsilon) \geq 0$ as asserted in Theorem 3.1 of \cite{casella1981estimating}. 
    In contrast, the Huber modification of the Bayes rule,
    \[
	\delta^H(x) = 
	\begin{cases} 
	    x + k & x < -b\\
	    \delta^B(x) & |x| < b\\
	    x -k & x > b
	\end{cases} 
    \]
    with $k = -(\log f_0)'(b)$ and $b$ solves $\int_{-b}^b f_0(x) dx + \frac{2 f_0(b)}{k} = (1-\epsilon)^{-1}$, has strictly greater risk for all $\theta \in \Theta = [-1,1]$.  

    To see this, let
    $f_0(x) = \frac{1}{2} \varphi(x-1) + \frac{1}{2} \varphi(x+1)$,
    $g^H(x) = \delta^H(x) - x$ and $g^B(x)  = \delta^B(x) - x = f_0'(x)/f_0(x)$. 
    By Stein's lemma, 
	\begin{align*}
		R(\delta^H ,\theta) & - R( \delta^B , \theta)\\  
		    & = 2 \E_\theta[ (g^H(X))' - (g^B(X))']\\ 
		    & + \E_\theta [ (g^H(X))^2 - (g^B(X))^2]\\
		& = -2 \E_\theta[ 1\{|X|>b\} (\log f_0(X))'' ]\\
		& \quad + \E_\theta [ 1\{|X|>b\} (k^2 - ((\log f_0(X))')^2]
	\end{align*}
    The first term is positive because $f_0(x)$ is log-concave. 
    The second term is also positive because $k = -(\log f_0)'(b)$ and 
    $k^2 >( (\log f_0(x))')^2$ for $|x|>b$, hence we conclude $\delta^H$ is inadmissible. 
    The Bayes rule and its Mallows modification are identical with risk slightly less than
    0.5, the Huber rule has pointwise risk greater than 0.62.

\item[Two point $G_0$ II]  When the two points of support of $G_0$ are more widely separated, for example,
     $G_0 = \frac{1}{2} \delta_2 + \frac{1}{2} \delta_{-2}$, and the support of $\theta$ is the whole 
     real line, the marginal $ f_0(x) = \int \varphi(x-\theta) dG_0(\theta)$ is bimodal, the Huber 
     least favaroable density for $X$ is described in the following proposition whose proof appears in
     Appendix \ref{app.A}.
     \begin{proposition} \label{prop:Huber}
	For $\epsilon$ sufficiently large,\footnote{For very small $\epsilon$ no modification
	in the center of the distribution is required; only the tail behavior is modified.
	In our example this threshold is about $\epsilon_0 < 0.0001$.}
	the solution to,
	\[
	\min_{F \in \mathcal{F}_\epsilon} I(F) 
	\]
	with $\mathcal{F}_\epsilon = \{F: F= (1-\epsilon) F_0 + \epsilon W\}$ where 
	$F_0(x) = \frac{1}{2} \Phi(x+2) + \frac{1}{2} \Phi(x-2)$ and $W$ is arbitrary 
	distribution of $X$. The least favorable distribution has its density of the form: 
	\[
	f^*(x) = \begin{cases} 
		(1-\epsilon) f_0(x) e^{k(x+b)} & x \leq -b\\
		(1-\epsilon) f_0(x) & |x| \in [c,b]\\
		A^2 \cosh^2(kx/2) & x \in [-c,c]\\
		(1-\epsilon) f_0(x) e^{-k(x-b)} & x \geq b
		\end{cases} 
	\]
	where for a given $k$, and $(b,c,A)$, 
	\begin{align*}
		k &= -(\log f_0)'(b)\\
		k \cdot \tanh(kc/2) & = (\log f_0)'(c)\\
		A^2 \cosh^2(kc/2) &= (1-\epsilon) f_0(c).
	\end{align*}
	The constant $k$ is defined implicitly by,
	\begin{align*}
	\int_{-c}^c A^2 & \cosh^2(kx/2) dx + 2 \int_c^b (1-\epsilon) f_0(x)dx\\ 
	& + \frac{2(1-\epsilon) f_0(b)}{k} = 1.
	\end{align*}
	The associated Huber decision rule is then, 
	\[
	\delta^*(x) = 
	\begin{cases} 
		x + k \cdot \tanh(kx/2) & x\in [-c,c]\\
		x + f_0(x)'/f_0(x) & |x| \in [c,b]\\
		x-k & x \geq b\\
		x+k & x \leq -b
	\end{cases} 
	\]
	\end{proposition}

	In Figure \ref{fig.2ptrule} we compare this shrinkage rule with the initial Bayes rule
	and the Hodges-Lehmann rule (labeled as Mallows) in the left panel.  Again we
	see that the Mallows's rule oscillates around the Huber rule in the tails, while being
	somewhat smoother in the around zero.  In the right panel of Figure \ref{fig.2ptrisk} we
	depict the pointwise risk of the three procedures for the choice, $\epsilon = 0.2$.
	Worst case risk for the Huber and Mallows rules is about 1.67 while the worst case 
	risk of the Bayes rule is unbounded. The Huber rule is now
	clearly admissible, but still not as attractive as the Mallows rule, in the sense that 
	its Bayes risk is strictly larger than the Mallows rule.

\begin{figure}
    \begin{center}
	\resizebox{.4\textwidth}{!}{{\includegraphics{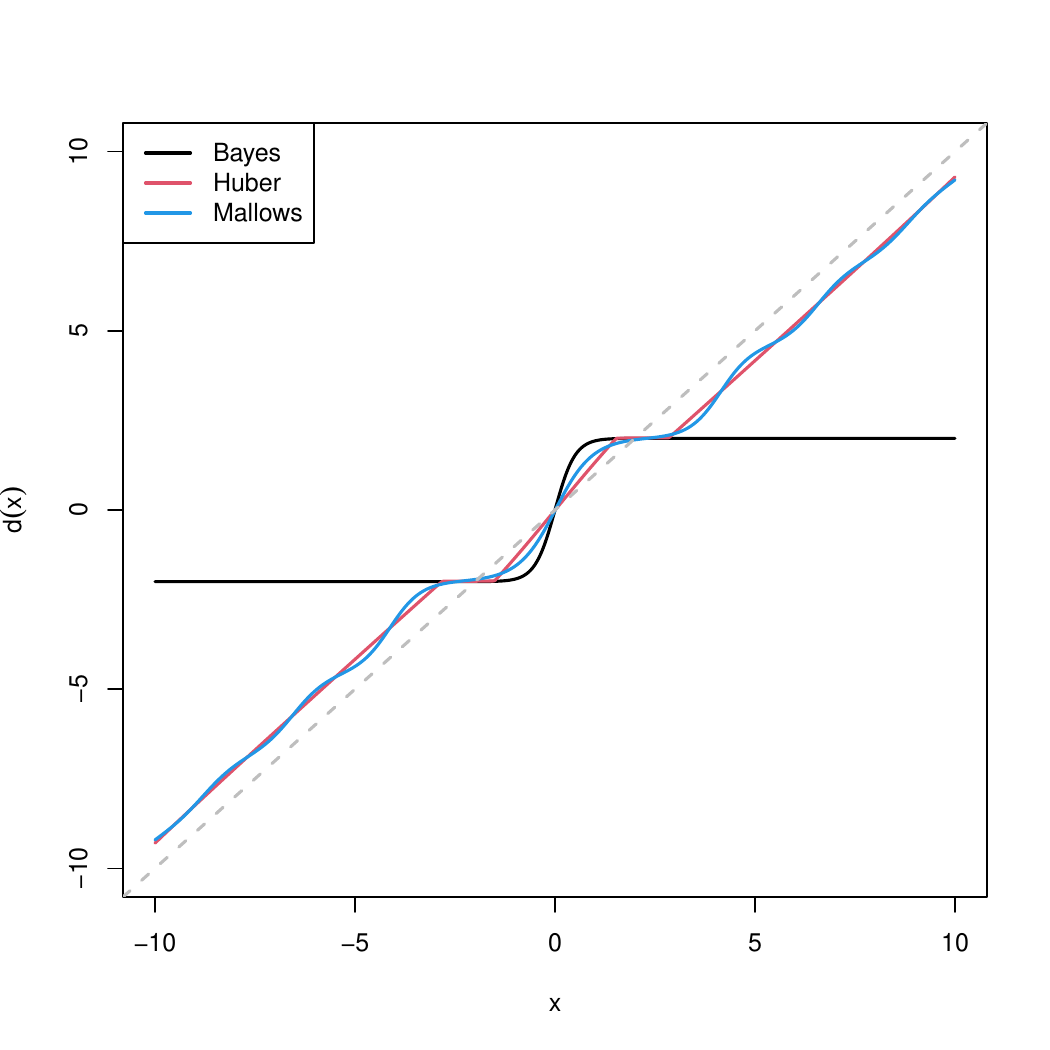}}}
    \end{center}
\caption{Decision rules for three procedures.
The initial prior is: $G_0 = \frac{1}{2} \delta_2 + \frac{1}{2} \delta_{-2}$,
with contamination level $\epsilon = 0.2$.  The worst case risk for 
Mallows and Huber rules is about 1.67. }\label{fig.2ptrule}
\end{figure}

\begin{figure}
    \begin{center}
	\resizebox{.4\textwidth}{!}{{\includegraphics{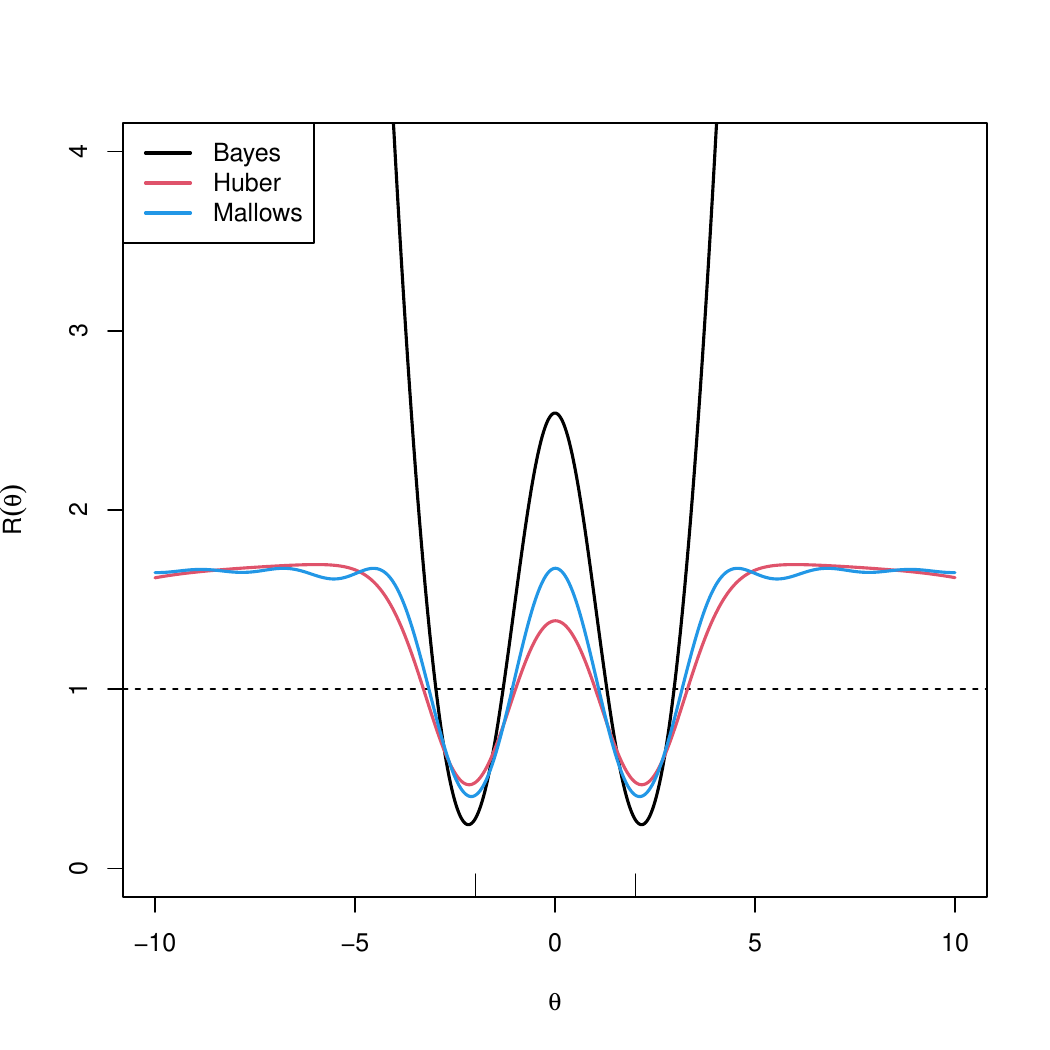}}}
    \end{center}
\caption{Associated pointwise risk functions for three procedures for the
initial prior: $G_0 = \frac{1}{2} \delta_2 + \frac{1}{2} \delta_{-2}$,
with contamination level $\epsilon = 0.2$.  The worst case risk for 
Mallows and Huber rules is about 1.67. }\label{fig.2ptrisk}
\end{figure}

\end{description}

\subsection{Restricted Empirical Bayes rules}
Having examined several examples of the Hodges and Lehmann restricted Bayes rules for some simple initial
prior distributions, we now consider an empirical Bayes counterpart with $G_0$ estimated by maximum
likelihood as proposed by \cite{kw} and \cite{km}, but already anticipated by \cite{r50}.
Given a sample from the compound decision problem posed in the introduction, we consider
the nonparametric maximum likelihood estimator $\hat G$ as the initial $G_0$ and then
proceed to construct a modified prior according to the principles laid out by Hodges and Lehmann.
    \begin{figure}
    \begin{center}
          \resizebox{.4\textwidth}{!}{{\includegraphics{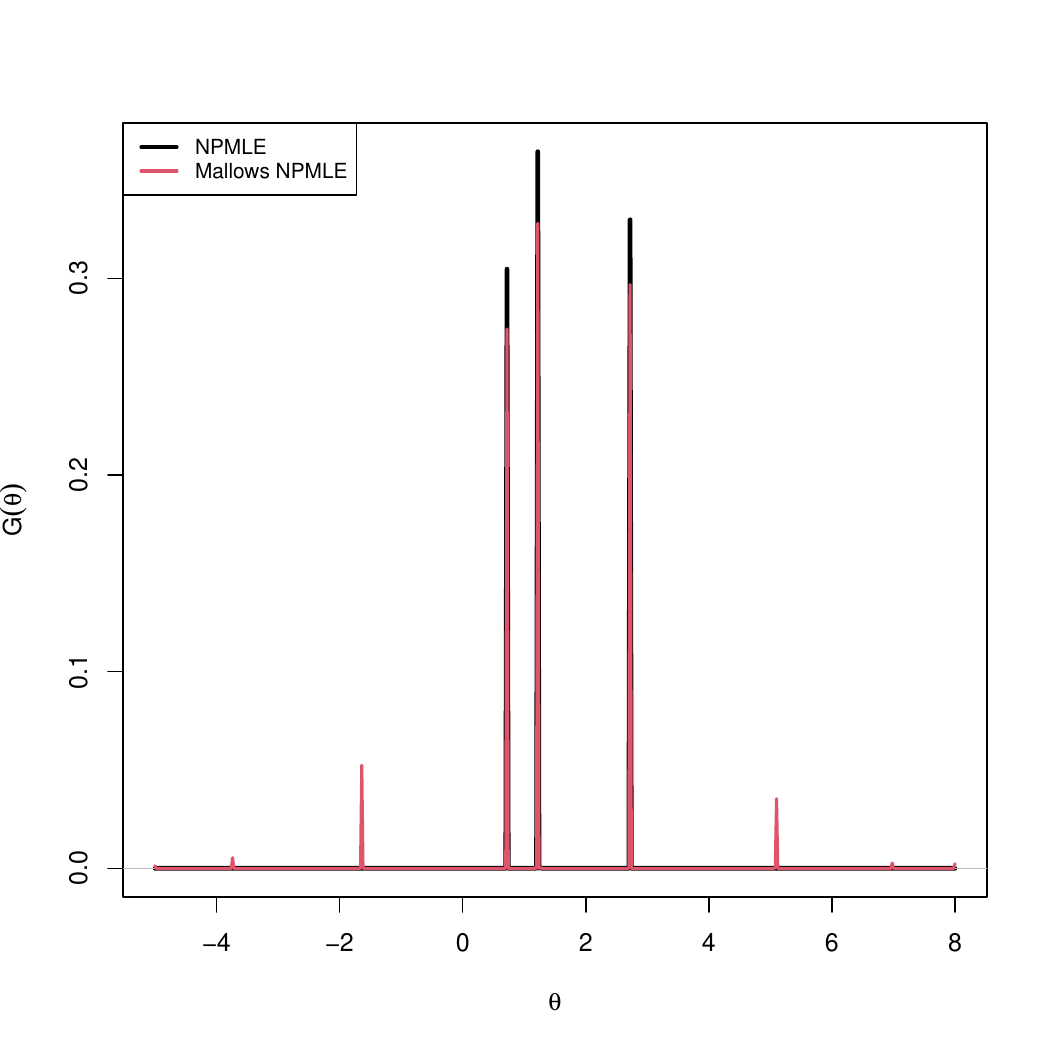}}}
    \end{center} 
    \caption{The figure contrasts the NPMLE prior with the modified
    Mallows prior.  In this example we take  $\epsilon = 0.1$ on the presumption that the initial
    empirical prior is more reliable than in the earlier examples.} \label{fig.npmle1}
    \end{figure}
    \begin{figure}
    \begin{center}
          \resizebox{.4\textwidth}{!}{{\includegraphics{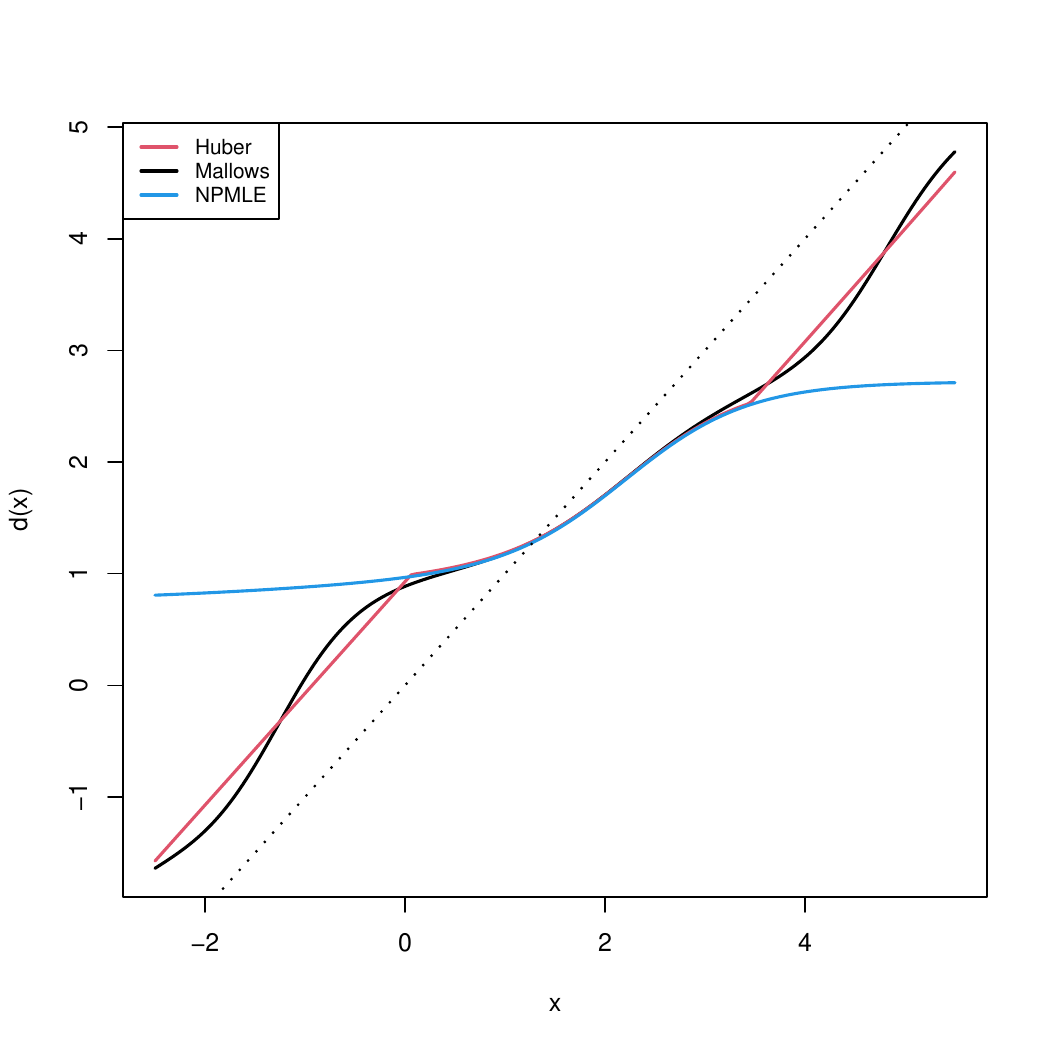}}}
    \end{center} 
    \caption{The figure contrasts decision rules for Huber and Mallows modifications
    of the NPMLE empirical Bayes rules.  The Huber rule is almost linear in the center
    while the Mallows rule is smoother near zero and oscillates around the Huber rule in
    the tails.} \label{fig.npmle2} 
    \end{figure}

    We illustrate the consequences of this in Figures \ref{fig.npmle1} and \ref{fig.npmle2}.  Data is generated from the 
    standard Gaussian sequence model with $G_0 \sim U[0,3]$.  Heavy black vertical lines indicate the 
    original NPMLE $\hat G$ while the red vertical lines indicate the mass points of the modified 
    prior.  While some alteration of the mass in the center of the estimated mixing distribution can 
    be seen, the main change is the new mass points in the tails which decline exponentially in 
    accordance with the Mallows's conjecture.  This feature resembles the \cite{efronmorris} limited 
    translation estimator that imposed linear shrinkage in the center of the distribution, but eschewed 
    shrinkage in the tails.  In baseball terms: a few extremely good hitters deserve their exalted averages.  
    Note that the restricted and unrestricted prior decision rules illustrated in the right panel of the 
    figure agree quite closely on the support of the true $\theta$'s, but diverge sharply beyond this support.

Denote the 
oracle Mallows rule $\delta^M$ with known $G_0$ and denote the empirical Mallows rule $\hat \delta^M$ with 
estimated (NPMLE) $\hat G$. In the following theorem, we establish oracle-risk consistency of the empirical Mallows rule. We also show that the worst case point-wise risk of the empirical Mallows rule converges to that of the oracle Mallows rule. 
The proof appears in Appendix \ref{app.B}.

\begin{theorem}\label{thm: EBMallows}
	Fix $\epsilon \in (0,1)$. Let $\hat G$ be an estimator of $G_0$ constructed from $X_1, \dots, X_n$. For any rule $\delta_n$, let the compound risk be 
	\[
	\mathcal{R}_n(\delta_n) = \mathbb{E}_{G_0} \Big[ \frac{1}{n} \sum_{i=1}^n \{\delta_n(X_i) - \theta_i^2\}\Big].
	\] 
	If $\hat G$ converges weakly to $G_0$ in probability as $n \to \infty$, then 
	\[
\left|
	\mathcal R_n^{\mathrm{ss}}(\hat\delta^M)
	-\mathcal R_n^{\mathrm{ss}}(\delta^M)
	\right|\longrightarrow0, 
	\]
	\[
	\left|\sup_{\theta \in \mathbb{R}}R(\hat \delta_n^M, \theta) - \sup_{\theta \in \mathbb{R}}R(\delta^M,\theta)\right| \overset{p}{\longrightarrow}0.
	\]
\end{theorem} 

\subsection{Some simulation experience} \label{sec: simu}
To evaluate the cost of imposing restrictions on the prior of the Hodges-Lehmann type we consider
three examples in this section:  
\begin{description}
    \item[Gaussian $G_0$]  $G_0 \sim \NN (0,1)$.
    \item[Uniform $G_0$] $G_0 \sim U[-2,2]$.
    \item[Twopoint $G_0$] $G_0 \sim 0.5(\delta_{-2} + \delta_2)$.
\end{description}
For each of these settings we compute mean squared error (MSE) for each of the following 
decision rules: 
\begin{description}
    \item[MLE]  Minimax Rule.
    \item[$\delta^L$]  Best Linear Rule.
    \item[$\delta^B$]  Bayes Rule.
    \item[$\hat \delta_{\hat G}^B$]  Bayes Rule with NPMLE $\hat G$.
    \item[$\delta^H$]  Huber Modified Bayes Rule.
    \item[$\delta^M$]  Mallows Modified Bayes Rule.
    \item[$\hat \delta_{\hat G}^H$]  Huber Modified Bayes Rule with NPMLE $\hat G$.
    \item[$\hat \delta_{\hat G}^M$]  Mallows Modified Bayes Rule with NPMLE $\hat G$.
\end{description}
The last four rules are evaluated for four distinct values of $\epsilon \in \{ 0.05, 0.1, 0.2, 0.4\}$.  
All the simulations are based on 500 replications, for each compound decision problem.

The linear rules work reasonably well for the Gaussian and Uniform settings, however they
perform poorly in the two-point setting.  The cost of the Hodges-Lehmann restricted priors
is modest for small $\epsilon$, but not surprisingly grows substantially when $\epsilon$ is
larger.  With only $n = 100$ observations, the NPMLE rule, $\delta_{\hat G}^B$, is too variable,
but for the larger sample sizes it is nearly competitive with the (oracle) Bayes rules.

For each $G_0$ we can evaluate for any $\epsilon$ the corresponding worst case pointwise risk of each rule. 
The minimax rule achieves a worst case pointwise risk of one, while the Bayes rule typically has unbounded 
pointwise risk. For the Mallows rule, this can be evaluated by $\E_{\theta^*}[(\delta^M(X)-\theta^*)^2]$ 
where $\theta^*$ is any mass point of $H^*$ as discussed in Section \ref{sec: computation}. 
For the Huber rule, we can show that for all the  $G_0$ we considered, 
$\sup_\theta R(\delta^H, \theta) = 1+k_\epsilon^2$ with $k_\epsilon = \sup_{x} | (\log f_\epsilon^H(X))'|$ 
in which $f_\epsilon^H$ is the least favorable Huber density for the given $G_0$ and $\epsilon$. 
Details along with some simulation results appear in Appendix \ref{app.C}. 

\begin{table}[!tbp]
\begin{center}
\begin{tabular}{lrrrrrrrr}
\hline\hline
\multicolumn{1}{l}{}&\multicolumn{1}{c}{MLE}&\multicolumn{1}{c}{$\delta^B$}&\multicolumn{1}{c}{$\hat \delta^B$}&\multicolumn{1}{c}{$\delta_{\hat G}^B$}&\multicolumn{1}{c}{$\delta^H$}&\multicolumn{1}{c}{$\delta^M$}&\multicolumn{1}{c}{$\delta_{\hat G}^H$}&\multicolumn{1}{c}{$\delta_{\hat G}^M$}\tabularnewline
\hline
{\bfseries n =  100}&&&&&&&&\tabularnewline
~~$\epsilon = $ 0.05&$1$&$0.50$&$0.52$&$0.57$&$0.53$&$0.53$&$0.56$&$0.57$\tabularnewline
~~$\epsilon = $ 0.1&$1$&$0.50$&$0.51$&$0.56$&$0.55$&$0.55$&$0.58$&$0.58$\tabularnewline
~~$\epsilon = $ 0.2&$1$&$0.50$&$0.52$&$0.57$&$0.60$&$0.60$&$0.62$&$0.62$\tabularnewline
~~$\epsilon = $ 0.4&$1$&$0.49$&$0.51$&$0.56$&$0.69$&$0.68$&$0.69$&$0.69$\tabularnewline
\hline
{\bfseries n =  500}&&&&&&&&\tabularnewline
~~$\epsilon = $ 0.05&$1$&$0.50$&$0.50$&$0.52$&$0.53$&$0.53$&$0.53$&$0.53$\tabularnewline
~~$\epsilon = $ 0.1&$1$&$0.50$&$0.51$&$0.52$&$0.56$&$0.55$&$0.56$&$0.56$\tabularnewline
~~$\epsilon = $ 0.2&$1$&$0.50$&$0.50$&$0.52$&$0.60$&$0.60$&$0.60$&$0.60$\tabularnewline
~~$\epsilon = $ 0.4&$1$&$0.50$&$0.50$&$0.52$&$0.69$&$0.69$&$0.69$&$0.69$\tabularnewline
\hline
{\bfseries n =  1000}&&&&&&&&\tabularnewline
~~$\epsilon = $ 0.05&$1$&$0.50$&$0.50$&$0.51$&$0.53$&$0.53$&$0.53$&$0.53$\tabularnewline
~~$\epsilon = $ 0.1&$1$&$0.50$&$0.50$&$0.51$&$0.55$&$0.55$&$0.55$&$0.55$\tabularnewline
~~$\epsilon = $ 0.2&$1$&$0.50$&$0.50$&$0.51$&$0.60$&$0.60$&$0.60$&$0.60$\tabularnewline
~~$\epsilon = $ 0.4&$1$&$0.50$&$0.50$&$0.51$&$0.69$&$0.68$&$0.69$&$0.69$\tabularnewline
\hline
\end{tabular}
\caption{Mean squared error for several compound decision rules. 
	 Data generated as $Y = \theta+ U$ with $\theta \sim \NN (0,1)$ and
	 $U \sim \NN (0,1)$, and initial prior is $G_0 = \NN (0,1)$.\label{tab.sim1a}}\end{center}
\end{table}

\begin{table}[!tbp]
\begin{center}
\begin{tabular}{lrrrrrrrr}
\hline\hline
\multicolumn{1}{l}{}&\multicolumn{1}{c}{MLE}&\multicolumn{1}{c}{$\delta^L$}&\multicolumn{1}{c}{$\delta^B$}&\multicolumn{1}{c}{$\delta_{\hat G}^B$}&\multicolumn{1}{c}{$\delta^H$}&\multicolumn{1}{c}{$\delta^M$}&\multicolumn{1}{c}{$\delta_{\hat G}^H$}&\multicolumn{1}{c}{$\delta_{\hat G}^M$}\tabularnewline
\hline
{\bfseries n =  100}&&&&&&&&\tabularnewline
~~$\epsilon = $ 0.05&$1.00$&$0.57$&$0.55$&$0.62$&$0.59$&$0.59$&$0.63$&$0.63$\tabularnewline
~~$\epsilon = $ 0.1&$0.99$&$0.57$&$0.55$&$0.61$&$0.62$&$0.61$&$0.65$&$0.65$\tabularnewline
~~$\epsilon = $ 0.2&$1.01$&$0.57$&$0.55$&$0.62$&$0.68$&$0.67$&$0.70$&$0.70$\tabularnewline
~~$\epsilon = $ 0.4&$1.01$&$0.57$&$0.55$&$0.62$&$0.76$&$0.76$&$0.77$&$0.78$\tabularnewline
\hline
{\bfseries n =  500}&&&&&&&&\tabularnewline
~~$\epsilon = $ 0.05&$1.00$&$0.57$&$0.55$&$0.57$&$0.59$&$0.59$&$0.60$&$0.60$\tabularnewline
~~$\epsilon = $ 0.1&$1.00$&$0.57$&$0.55$&$0.57$&$0.63$&$0.62$&$0.63$&$0.62$\tabularnewline
~~$\epsilon = $ 0.2&$1.00$&$0.57$&$0.55$&$0.57$&$0.67$&$0.67$&$0.68$&$0.67$\tabularnewline
~~$\epsilon = $ 0.4&$1.00$&$0.57$&$0.55$&$0.57$&$0.76$&$0.76$&$0.76$&$0.76$\tabularnewline
\hline
{\bfseries n =  1000}&&&&&&&&\tabularnewline
~~$\epsilon = $ 0.05&$1.00$&$0.57$&$0.55$&$0.56$&$0.59$&$0.59$&$0.60$&$0.59$\tabularnewline
~~$\epsilon = $ 0.1&$1.00$&$0.57$&$0.55$&$0.56$&$0.62$&$0.62$&$0.63$&$0.62$\tabularnewline
~~$\epsilon = $ 0.2&$1.00$&$0.57$&$0.55$&$0.56$&$0.67$&$0.67$&$0.68$&$0.67$\tabularnewline
~~$\epsilon = $ 0.4&$1.00$&$0.57$&$0.55$&$0.56$&$0.76$&$0.76$&$0.76$&$0.76$\tabularnewline
\hline
\end{tabular}
\caption{Mean squared error for several compound decision rules. 
	 Data generated as $Y = \theta+ U$ with $\theta \sim U (-2,2)$ and
	 $U \sim \NN (0,1)$, and initial prior is $G_0 = U(-2,2)$.\label{tab.sim2a}}\end{center}
\end{table}

\begin{table}[!tbp]
\begin{center}
\begin{tabular}{lrrrrrrrr}
\hline\hline
\multicolumn{1}{l}{}&\multicolumn{1}{c}{MLE}&\multicolumn{1}{c}{$\delta^L$}&\multicolumn{1}{c}{$\delta^B$}&\multicolumn{1}{c}{$\delta_{\hat G}^B$}&\multicolumn{1}{c}{$\delta^H$}&\multicolumn{1}{c}{$\delta^M$}&\multicolumn{1}{c}{$\delta_{\hat G}^H$}&\multicolumn{1}{c}{$\delta_{\hat G}^M$}\tabularnewline
\hline
{\bfseries n =  100}&&&&&&&&\tabularnewline
~~$\epsilon = $ 0.05&$1.00$&$0.80$&$0.27$&$0.33$&$0.33$&$0.31$&$0.37$&$0.36$\tabularnewline
~~$\epsilon = $ 0.1&$0.99$&$0.80$&$0.28$&$0.34$&$0.38$&$0.34$&$0.42$&$0.39$\tabularnewline
~~$\epsilon = $ 0.2&$1.01$&$0.81$&$0.29$&$0.35$&$0.49$&$0.42$&$0.53$&$0.47$\tabularnewline
~~$\epsilon = $ 0.4&$1.01$&$0.81$&$0.28$&$0.34$&$0.68$&$0.59$&$0.70$&$0.63$\tabularnewline
\hline
{\bfseries n =  500}&&&&&&&&\tabularnewline
~~$\epsilon = $ 0.05&$1.00$&$0.80$&$0.28$&$0.29$&$0.33$&$0.32$&$0.34$&$0.33$\tabularnewline
~~$\epsilon = $ 0.1&$1.00$&$0.80$&$0.27$&$0.29$&$0.38$&$0.34$&$0.39$&$0.35$\tabularnewline
~~$\epsilon = $ 0.2&$1.00$&$0.80$&$0.27$&$0.29$&$0.48$&$0.41$&$0.50$&$0.42$\tabularnewline
~~$\epsilon = $ 0.4&$1.00$&$0.81$&$0.28$&$0.29$&$0.67$&$0.58$&$0.68$&$0.60$\tabularnewline
\hline
{\bfseries n =  1000}&&&&&&&&\tabularnewline
~~$\epsilon = $ 0.05&$1.00$&$0.80$&$0.27$&$0.28$&$0.32$&$0.31$&$0.33$&$0.32$\tabularnewline
~~$\epsilon = $ 0.1&$1.00$&$0.80$&$0.27$&$0.28$&$0.38$&$0.34$&$0.39$&$0.35$\tabularnewline
~~$\epsilon = $ 0.2&$1.00$&$0.80$&$0.27$&$0.28$&$0.48$&$0.41$&$0.49$&$0.42$\tabularnewline
~~$\epsilon = $ 0.4&$1.00$&$0.80$&$0.27$&$0.28$&$0.67$&$0.58$&$0.67$&$0.59$\tabularnewline
\hline
\end{tabular}
\caption{Mean squared error for several compound decision rules. 
	 Data generated as $Y = \theta+ U$ with $\theta\sim 0.5 (\delta_{-2} + \delta_2)$ and
	 $U \sim \NN (0,1)$, and initial prior is $G_0 = 0.5 (\delta_{-2} + \delta_2)$.\label{tab.sim3a}}\end{center}
\end{table}

\section{Robustified Gaussian Likelihoods}

Rather than robustifying the  prior an alternative strategy is to robustify the
likelihood. 
We will consider two variants of this: the first
following \cite{huber64} and the second following \cite{mallows78}.

The classical procedure of Huber for estimating a location parameter is easily adapted
to the Gaussian sequence compound decision problem.  In place of the Gaussian likelihood
in the NPMLE problem we simply insert the Huber log likelihood with density,
\[
\varphi (u) = 
\begin{cases}
    (1-\epsilon) \varphi (k) \exp(-k(u-k)) & u > k\\
    (1-\epsilon)\varphi(u) & |u| \leq k\\
    (1-\epsilon) \varphi (k)\exp(k(u+k)) & u < -k
\end{cases}
\]
where $\epsilon$ and $k$ are linked by $2 \varphi(k)/k - 2 \Phi(-k) = \epsilon/(1-\epsilon)$.
This density is least favorable, that is has minimal Fisher information for location, 
in the contamination model,
\[
\Psi_\epsilon= \{ \Psi =  (1 - \epsilon) \Phi + \epsilon H \},
\]
over all symmetric distributions $H$.  When $\epsilon = 1/2$ the least favorable
Huber distribution is Laplace, or double exponential, and can be viewed as least favorable
against asymmetric noise as well as symmetric.

\cite{mallows78} proposes to consider minimizing $I(\Phi * G)$ over $\GG$, the set of all 
distributions with mass $1 - \epsilon$ at zero, provides an alternative to the Huber
contamination model. Rather than assuming iid innovations each arising from the Huber
mixture model, Mallows considers an additive outlier model in which with probability $\epsilon$
innovations are standard Gaussian, but occasionally are generated by the convolution $\Phi * H$. 

\begin{figure}
  \begin{center}
          \resizebox{0.4\textwidth}{!}{{\includegraphics{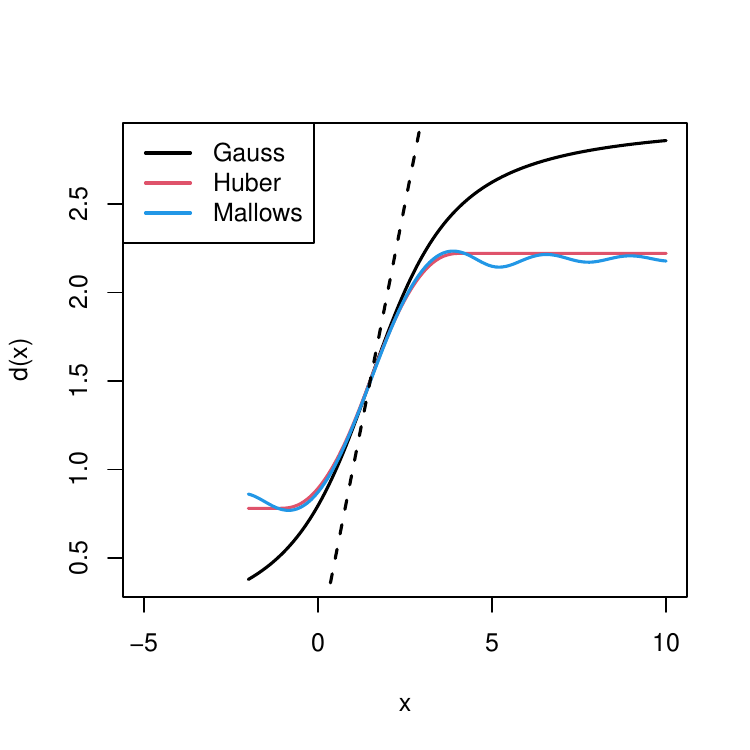}}}
  \end{center} 
  \caption{The Huber and Mallows decision rules are contrasted with the Gaussian rule
  in a setting with $G = U[0,3]$.  The heavier tail behavior of the Huber and Mallows
  base distribution, $\varphi$ results in more aggressive shrinkage with extreme 
  observations discounted as the consequence of noise rather than signal.  The Huber
  and Mallows rules both set $\epsilon = 0.1$ for this figure.}
  \label{fig.rules}
\end{figure}

In Figure \ref{fig.rules} we contrast the Huber and Mallows decision rules with the traditional
Gaussian rule.  The true mixing distribution, $G$, is chosen to be $U[0,3]$ so the Gaussian rule is
itself somewhat curved, not the linear rule we would expect were $G$ itself Gaussian.  In contrast
the Huber and Mallows rules impose a more aggressive form of shrinkage.  With Gaussian $\varphi$
extreme observations can be confidently attributed to signal, while the heavier tailed $\varphi$
of the Huber and Mallows rules tend to attribute such observations to noise.  As we have seen
previously, the Mallows rule oscillates around the Huber rule in the tails, but otherwise their
behavior is quite similar.

When the usual Gaussian $\varphi$ is replaced by either the Huber or Mallows alternative in the
nonparametric maximum likelihood estimation of $G$ solutions are accordingly more concentrated
with fewer extreme mass points.  This effect accentuates the more aggressive shrinkage effect
observed in Figure \ref{fig.rules}.

Both the Mallows and Huber least favorable contamination models offer principled alternatives to the
strictly Gaussian noise model.  They preserve the convexity of the underlying NPMLE 
problem and therefore can be easily implemented in software.  In the next section we
compare performance of several variants of these procedures for a few simulated compound decision
settings.

\subsection{Some Simulation Experience}

We consider the compound decision problem with observations generated from,
\[
Y_i = \theta_i + U_i, \quad i = 1, \dots, n,
\]
with $\theta_i$ and $U_i$ independent and each generated iidly from $G$ and $\Psi$ respectively.
There are two choices of $G$:  Either $G \sim U[0,3]$ or $G \sim 0.9 \delta_0 + 0.1 \delta_3$.
And three choices of $\Psi$:  $\Psi \sim 0.8 \Phi + 0.2 \Phi(\cdot / 3)$, $\Psi \sim \text{Laplace}$
and $\Psi = \Phi$, which we label Tukey, Laplace and Gauss respectively.

In Table \ref{tab.sim} we compare mean squared error performance of ten options with an infeasible oracle
procedure that ``knows'' both the $\Psi$ and $G$ distributions.  The experiment has 500 replications each
with sample size $n = 500$.  The competing feasible decision rules are:  GLmix,  the Gaussian NPMLE;
Laplace,  the Laplacian NPMLE; HL$(\epsilon)$, the Huber NPMLE; ML$(\epsilon)$,  the Mallows NPMLE, 
with $\epsilon \in \{ 0.20, 0.10, 0.05, 0.025 \}$.

\begin{table}[!tbp]
\caption{Mean Squared Error of Several Compound Decision Rules\label{tab.sim}} 
\centering
\begingroup
\setlength{\tabcolsep}{4pt}
\begin{tabular}{@{}lrrrcrrr@{}}
\hline\hline
\multicolumn{1}{l}{\bfseries }&\multicolumn{3}{c}{\bfseries $G \sim 0.9 \delta_0 + 0.1 \delta_3$}&\multicolumn{1}{c}{\bfseries }&\multicolumn{3}{c}{\bfseries $G  \sim U[0,3]$}\tabularnewline
\cline{2-4} \cline{6-8}
\multicolumn{1}{l}{}&\multicolumn{1}{c}{Tukey}&\multicolumn{1}{c}{Laplace}&\multicolumn{1}{c}{Gauss}&\multicolumn{1}{c}{}&\multicolumn{1}{c}{Tukey}&\multicolumn{1}{c}{Laplace}&\multicolumn{1}{c}{Gauss}\tabularnewline
\hline
Oracle&$0.50$&$0.33$&$0.34$&&$0.47$&$0.38$&$0.42$\tabularnewline
GLmix&$0.96$&$0.42$&$0.35$&&$1.00$&$0.53$&$0.44$\tabularnewline
LLmix&$0.69$&$0.40$&$0.46$&&$0.63$&$0.48$&$0.48$\tabularnewline
HL(0.2)&$0.70$&$0.58$&$0.66$&&$0.66$&$0.71$&$0.67$\tabularnewline
HL(0.1)&$0.69$&$0.40$&$0.43$&&$0.69$&$0.56$&$0.52$\tabularnewline
HL(0.05)&$0.78$&$0.36$&$0.37$&&$0.79$&$0.50$&$0.46$\tabularnewline
HL(0.025)&$0.86$&$0.37$&$0.35$&&$0.88$&$0.49$&$0.44$\tabularnewline
ML(0.2)&$0.67$&$0.57$&$0.64$&&$0.61$&$0.68$&$0.61$\tabularnewline
ML(0.1)&$0.66$&$0.39$&$0.42$&&$0.65$&$0.54$&$0.50$\tabularnewline
ML(0.05)&$0.75$&$0.36$&$0.37$&&$0.76$&$0.48$&$0.46$\tabularnewline
ML(0.025)&$0.84$&$0.36$&$0.35$&&$0.86$&$0.48$&$0.44$\tabularnewline
\hline
\end{tabular}
\endgroup
\end{table}

It is evident from the table that the Gaussian NPMLE, GLmix, performs best when
the noise distribution is actually Gaussian; however, when $\Psi \neq \Phi$
it pays to consider one of the alternatives.  The Mallows NPMLE procedures seem
to perform slightly better than the corresponding Huber methods, while the Laplace
NPMLE, LLmix, which can be regarded as a ``median-type'' estimator performs surprisingly
well over all the experimental settings.  

\section{conclusion}
We have considered two distinct strategies for robustifying empirical Bayes decision rules for the
Gaussian sequence model.  In the first motivated by the seminal paper of Hodges and Lehmann we would
like protection against deviations from an initial Bayes prior.  In the second we seek 
protection against non-Gaussian behavior in the noise distribution. Both strategies rely on
the classical robustness proposals of \cite{huber64} and \cite{mallows78}.  Some combination of
the two strategies is obviously possible, but choice of tuning parameters remains a delicate issue.

\section*{Supplementary Material}
\appendix
\section{Proof of Proposition 3.1} 
\begin{proof}
For all $F \in \mathcal{F}_\epsilon$, we can represent its density as $f(x) = (1-\epsilon) f_0(x) + \epsilon h(x)$.	
When $f_0$ is symmetric around 0, which is our case here, then the optimal $f^*$ will necessarily have $h$ symmetric and 
hence $f^*$ itself symmetric around 0. Let $u(x) = \sqrt{f(x)}$, then $I(F) = 4 \int (u(x)')^2 dx$ 
with the constraint $u^2(x) \geq (1-\epsilon) f_0(x)$ and $\int u^2(x)dx = 1$. The Lagrangean is then 
\begin{align*}
	L(u) & = 4 \int (u'(x))^2 dx + \lambda \Big (\int u(x)^2 dx  -1\Big )\\
	& + \int \eta(x) \{(1-\epsilon) f_0(x) - u(x)^2\}dx.
\end{align*}
Denote $\psi(x) = \sqrt{(1-\epsilon)f_0(x)}$, then the KKT conditions for an optimal solution 
$u^*(x) = \sqrt{f^*(x)}$ are:
	\begin{align*}
		(u^*(x))'' &= \frac{\lambda - \eta(x)}{4} u^*(x)\\
		\eta(x) (u^*(x) - \psi(x)) & = 0\\
		u^*(x)& \geq \psi(x)\\
		\eta(x)& \geq 0\\
		\lambda & \in \RR.
	\end{align*}
Define the two sets: 
\begin{align*}
	\mathcal{C} &= \{x: u^*(x) = \psi(x)\}\\
	\mathcal{F} & = \{x: u^*(x) > \psi(x)\}
\end{align*} 
	On set $\mathcal{F}$, $\eta(x) =0$, hence $(u^*(x))'' = \frac{\lambda}{4} u^*(x)$.
	The multiplier $\lambda$ associated with the normalization constraint is unrestricted a priori.
	Since $\mathcal{F}$ contains an unbounded tail component and $u^*$ is positive and square-integrable,
	the differential equation on that component forces $\lambda>0$. Write $k=\sqrt{\lambda}$.
	This differential equation has solution on each connected component of $\mathcal{F}$ of the form, 
\[
u^*(x) = A \exp(kx/2) + B \exp(-kx/2),
\]
	For any closed subset of $\mathcal{F}$, say $|x| < c$ 
both exponential pieces are involved and $u^*(x) = A \cosh(kx/2)$. 
For any subset in the tail we need to eliminate one of the exponential terms to ensure integrability. 
Hence for the right tail component, with $x > b$, $u^*(x) = \sqrt{(1-\epsilon) f_0(b)}\exp(-k(x-b)/2)$ 
and for the left tail component with 
$x < -b$, and $u^*(x) = \sqrt{(1-\epsilon) f_0(-b)} \exp(k(x+b)/2)$. 
The remaining task is to find $(b, c, k, A)$ to guarantee that $f^*$ is in the feasible set 
$\mathcal{F}_\epsilon$. Continuity of $u^*(x)$ at the cut off points $\pm c$, 
implies $u^*(c) = \psi(c)$ which leads to, 
\begin{equation}\label{id1}
A = \frac{\psi(c)}{\cosh(kc/2)}.
\end{equation} 
We also require $(u^*)'(c) = \psi'(c)$ and $(u^*)'(b) = \psi'(b)$, 
These restrictions imply, 
\begin{equation}\label{id2}
	k = -(\log f_0)'(b)
\end{equation}
\begin{equation} \label{id3}
k \cdot \tanh(kc/2)  = (\log f_0)'(c).
\end{equation}
Finally, $f^*$ has to integrate to 1, which implies,
\begin{align} \label{id4}
\int_{-c}^c A^2 \cosh^2(kx/2) & + 2\int_c^b (1-\epsilon)f_0(x) dx\\ 
& + \frac{2 (1-\epsilon) f_0(b)}{k} = 1.\nonumber
\end{align} 
Solving \eqref{id1} to \eqref{id4} yields the constants, $(b,c,k,A)$. 
\end{proof} 

\section{Proof of Theorem 3.2}
\label{app:EBMallows-proof}

Fix \(\epsilon\in(0,1)\).  For an initial prior
\(G\in\mathcal P(\RR)\), let
\[
\GG_\epsilon(G)
=\big\{(1-\epsilon)G+\epsilon H:H\in\mathcal P(\RR)\big\}
\]
and define
\[
J_G(\delta)
=(1-\epsilon)r(G,\delta)
+\epsilon\sup_{\theta\in\RR}R(\delta,\theta)
=\sup_{P\in\GG_\epsilon(G)}r(P,\delta).
\]

The existence of a proper least favorable prior for this Gaussian-mixture
problem follows from Example~(ii)(b) and Theorems~1 and~3 of
\cite{bickel1983minimizing}; the restricted-Bayes saddle-point
characterization is given by Theorems~1 and~2 of \cite{hodgeslehmann}.
Taken together, these results imply that there is
\(P_G^M\in\GG_\epsilon(G)\) whose Bayes rule
\[
\delta_G^M(x)
=x+\frac{(f_G^M)'(x)}{f_G^M(x)},
\qquad f_G^M=\varphi*P_G^M,
\]
satisfies
\begin{equation}
\label{eq:B-saddle}
r(P,\delta_G^M)
\le r(P_G^M,\delta_G^M)
\le r(P_G^M,\delta)
\end{equation}
for every \(P\in\GG_\epsilon(G)\) and every decision rule \(\delta\).
Consequently, \(\delta_G^M\) minimizes \(J_G\), and
\begin{equation}
\label{eq:B-value}
J_G(\delta_G^M)
=r(P_G^M,\delta_G^M)
\le r(P_G^M,\delta^{\dagger})=1.
\end{equation}
where $\delta^{\dagger}(X) = X$ is the identity rule. The oracle and empirical Mallows rules in Theorem 3.2 
are \(\delta^M=\delta_{G_0}^M\) and
\(\hat\delta^M=\delta_{\hat G_n}^M\), respectively.

We need four elementary lemmas.

\begin{lemma}
\label{lem:B-shrinkage}
For every \(G\in\mathcal P(\RR)\),
\[
\begin{aligned}
\sup_{\theta\in\RR}R(\delta_G^M,\theta)&\le\epsilon^{-1},\\
\sup_{x\in\RR}|\delta_G^M(x)-x|&\le K_\epsilon,
\qquad K_\epsilon:=\sqrt{2/\epsilon}.
\end{aligned}
\]
\end{lemma}

\begin{proof}
The identity rule $\delta^{\dagger}(X) = X$ has constant pointwise risk
\[
R(\delta^{\dagger},\theta)
=\E_\theta(X-\theta)^2=1,
\qquad \theta\in\RR.
\]
Consequently, its Bayes risk under every prior is also one.  Since
\(\delta_G^M\) minimizes \(J_G\), equation~\eqref{eq:B-value} can
equivalently be written as
\[
J_G(\delta_G^M)\le J_G(\delta^{\dagger})=1.
\]
On the other hand, the Bayes-risk term in \(J_G\) is nonnegative, so
\[
\epsilon\sup_{\theta\in\RR}R(\delta_G^M,\theta)
\le J_G(\delta_G^M)\le1.
\]
Dividing by \(\epsilon\) gives
\[
\sup_{\theta\in\RR}R(\delta_G^M,\theta)\le\epsilon^{-1}.
\]
A posterior mean in the Gaussian location model is nondecreasing because
\[
(\delta_G^M)'(x)
=\operatorname{Var}_{P_G^M}(\theta\mid X=x)\ge0.
\]
Fix \(x\in\RR\) and take \(X\sim\NN(x,1)\).  If
\(\delta_G^M(x)\ge x\), then on \(\{X\ge x\}\),
monotonicity gives
\[
\delta_G^M(X)-x
\ge\delta_G^M(x)-x\ge0.
\]
Hence the squared error on this event is at least
\(\{\delta_G^M(x)-x\}^2\).  If instead \(\delta_G^M(x)<x\), then on
\(\{X\le x\}\),
\[
x-\delta_G^M(X)
\ge x-\delta_G^M(x)>0,
\]
and the same lower bound for the squared error holds.  Each of the two
events has probability \(1/2\) under \(X\sim\NN(x,1)\).  Therefore,
regarding \(x\) here as the value of the parameter,
\[
R(\delta_G^M,x)
\ge\frac12\{\delta_G^M(x)-x\}^2.
\]
Combining the last two displays gives, for every \(x\),
\[
|\delta_G^M(x)-x|^2
\le2R(\delta_G^M,x)
\le2\sup_{\theta\in\RR}R(\delta_G^M,\theta)
\le\frac{2}{\epsilon}.
\]
Taking square roots and then the supremum over \(x\) proves the result.
\end{proof}

\begin{lemma}
\label{lem:B-convexity}
For decision rules \(\delta_1,\delta_2\) with finite criterion values,
\begin{align}
J_G\left(\frac{\delta_1+\delta_2}{2}\right)
&\le \frac{J_G(\delta_1)+J_G(\delta_2)}{2} \notag\\
&\quad-\frac{1-\epsilon}{4}
\int(\delta_1-\delta_2)^2f_G.                 \label{eq:B-convexity}
\end{align}
In particular, for every \(\delta\) with \(J_G(\delta)<\infty\),
\begin{equation}
\label{eq:B-distance}
\int(\delta-\delta_G^M)^2f_G
\le\frac{2}{1-\epsilon}
\{J_G(\delta)-J_G(\delta_G^M)\}.
\end{equation}
\end{lemma}

\begin{proof}
Let \(\bar\delta=(\delta_1+\delta_2)/2\).  Expanding the square gives
\begin{align*}
R(\bar\delta,\theta)
&=\frac{R(\delta_1,\theta)+R(\delta_2,\theta)}{2}\\
&\quad-\frac14\int
  (\delta_1-\delta_2)^2(x)\varphi(x-\theta)\,dx .
\end{align*}
Integrating this identity with respect to \(G\) gives
\[
r(G,\bar\delta)
=\frac{r(G,\delta_1)+r(G,\delta_2)}{2}
-\frac14\int(\delta_1-\delta_2)^2f_G.
\]
On the other hand,
\[
\sup_{\theta\in\RR}R(\bar\delta,\theta)
\le\frac12\sup_{\theta\in\RR}R(\delta_1,\theta)
+\frac12\sup_{\theta\in\RR}R(\delta_2,\theta),
\]
Multiplying the
last two displays by \(1-\epsilon\) and \(\epsilon\), respectively,
and adding proves \eqref{eq:B-convexity}.

To obtain \eqref{eq:B-distance}, set
\(\delta_1=\delta\) and \(\delta_2=\delta_G^M\) in
\eqref{eq:B-convexity}. Since \(\delta_G^M\) minimizes \(J_G\) and
\(\bar\delta=(\delta+\delta_G^M)/2\) is an admissible decision rule,
\[
J_G(\delta_G^M)\leq J_G(\bar\delta).
\]
Combining this inequality with \eqref{eq:B-convexity} gives
\[
J_G(\delta_G^M)
\leq \frac{J_G(\delta)+J_G(\delta_G^M)}{2}
-\frac{1-\epsilon}{4}
\int(\delta-\delta_G^M)^2f_G.
\]
Rearranging yields
\[
\int(\delta-\delta_G^M)^2f_G
\leq\frac{2}{1-\epsilon}
\{J_G(\delta)-J_G(\delta_G^M)\},
\]
which is \eqref{eq:B-distance}.
\end{proof}

\begin{lemma}
\label{lem:B-continuity}
For \(K>0\), let
\[
\mathcal D_K
=\{\delta:\delta(x)=x+g(x),\ \|g\|_\infty\le K\}.
\]
For any \(G,\widetilde G\in\mathcal P(\RR)\),
\begin{align}
\sup_{\delta\in\mathcal D_K}
|r(G,\delta)-r(\widetilde G,\delta)|
&\le K^2\|f_G-f_{\widetilde G}\|_1 \notag\\
&\quad+2K\|f_G'-f_{\widetilde G}'\|_1.       \label{eq:B-risk-cont}
\end{align}
Moreover, if \(G_n\Rightarrow G\), then
\[
\|f_{G_n}-f_G\|_1\longrightarrow0,
\qquad
\|f_{G_n}'-f_G'\|_1\longrightarrow0.
\]
\end{lemma}

\begin{proof}
For \(\delta(x)=x+g(x)\), Fubini's theorem and
\[
f_G'(x)
=-\int(x-\theta)\varphi(x-\theta)\,dG(\theta)
\]
give
\begin{equation}
\label{eq:B-risk-expansion}
r(G,\delta)
=1-2\int g f_G'+\int g^2f_G.
\end{equation}
The bound \eqref{eq:B-risk-cont} follows immediately.

Suppose next that \(G_n\Rightarrow G\).  For each \(x\), bounded
continuity of \(\theta\mapsto\varphi(x-\theta)\) and
\(\theta\mapsto\varphi'(x-\theta)\) gives
\[
f_{G_n}(x)\longrightarrow f_G(x),
\qquad
f_{G_n}'(x)\longrightarrow f_G'(x).
\]
Scheff\'e's lemma yields \(\|f_{G_n}-f_G\|_1\to0\).  For the
derivatives, put
\begin{align*}
q_n(x)&=\int|x-\theta|\varphi(x-\theta)\,dG_n(\theta),\\
q(x)&=\int|x-\theta|\varphi(x-\theta)\,dG(\theta).
\end{align*}
For fixed \(x\), the function
\[
\theta\longmapsto|x-\theta|\varphi(x-\theta)
\]
is continuous and bounded uniformly in \(\theta\). Thus weak convergence also gives
\(q_n(x)\to q(x)\).  Moreover, since the integrand defining \(q_n\)
is nonnegative, Tonelli's theorem gives
\begin{align*}
\int q_n(x)\,dx
&=\int\!\int |x-\theta|\varphi(x-\theta)\,dx\,dG_n(\theta)\\
&=\int\left\{\int |z|\varphi(z)\,dz\right\}dG_n(\theta),
\end{align*}
where the second equality uses the change of variable \(z=x-\theta\).
The inner integral does not depend on \(\theta\). Since \(G_n\) is a
probability measure,
\[
\int q_n(x)\,dx=c,
\qquad
c:=\int |z|\varphi(z)\,dz.
\]
The identical argument with \(G\) in place of \(G_n\) gives
\(\int q(x)\,dx=c\). Finally, using \(\varphi'(z)=-z\varphi(z)\),
\[
c=2\int_0^\infty z\varphi(z)\,dz
=2\varphi(0)=\sqrt{\frac{2}{\pi}}.
\]
Set \(a_n(x)=|f_{G_n}'(x)-f_G'(x)|\).  The pointwise derivative
convergence above shows that \(a_n(x)\to0\), while the triangle
inequality gives
\[
0\le a_n(x)\le q_n(x)+q(x).
\]
Consequently, \(q_n+q-a_n\ge0\) and
\(q_n+q-a_n\to2q\) pointwise.  Fatou's lemma now yields
\begin{align*}
2c
=\int 2q
&\le\liminf_{n\to\infty}
   \int(q_n+q-a_n)\\
&=2c-\limsup_{n\to\infty}\int a_n.
\end{align*}
It follows that \(\limsup_n\int a_n\le0\).  Since \(a_n\ge0\), this
proves
\[
\|f_{G_n}'-f_G'\|_1=\int a_n\longrightarrow0.
\]
\end{proof}

\begin{lemma}
\label{lem:B-local-uniform}
Let \(p\) be a continuous, strictly positive density on \(\RR\), let
\(u\) be continuous and nondecreasing, and let \(u_n\) be possibly random
nondecreasing functions.
If
\[
D_n:=\int (u_n-u)^2p\xrightarrow{p}0,
\]
then, for every \(0<M<\infty\),
\[
\sup_{|x|\le M}|u_n(x)-u(x)|\xrightarrow{p}0.
\]
\end{lemma}

\begin{proof}
Fix \(0<M<\infty\) and \(\eta>0\).  By uniform continuity of \(u\) on
\([-M-1,M+1]\), there is an \(h\in(0,1]\) such that
\(|u(y)-u(x)|\le\eta/2\) whenever \(x,y\in[-M-1,M+1]\) and
\(|x-y|\le h\).
Also set
\[
m_M=\inf_{|x|\le M+1}p(x)>0.
\]
Suppose first that \(u_n(x)-u(x)>\eta\) for some \(x\in[-M,M]\).
For \(y\in[x,x+h]\), monotonicity of \(u_n\) gives
\[
u_n(y)-u(y)
\ge u_n(x)-u(x)-|u(y)-u(x)|>\eta/2.
\]
If instead \(u_n(x)-u(x)<-\eta\), then for \(y\in[x-h,x]\),
\[
u_n(y)-u(y)
\le u_n(x)-u(x)+|u(y)-u(x)|<-\eta/2.
\]
In either case,
\[
D_n\ge \frac14m_Mh\eta^2.
\]
Consequently,
\[
\left\{\sup_{|x|\le M}|u_n(x)-u(x)|>\eta\right\}
\subseteq
\left\{D_n\ge\frac14m_Mh\eta^2\right\}.
\]
Since \(D_n\xrightarrow{p}0\), the probability of the event on the
right tends to zero. This proves the claim.
\end{proof}

\begin{proof}[Proof of Theorem 3.2]
Write
\[
f_0=f_{G_0},\qquad
\hat f_n=f_{\hat G_n},\qquad
\delta_0=\delta^M,\qquad
\delta_n=\hat\delta_n^M .
\]
By the assumed weak consistency of $\hat G_n$, Lemma~\ref{lem:B-continuity}, and the
continuous-mapping theorem, with \(K_\epsilon=\sqrt{2/\epsilon}\),
\[
\Delta_n
:=K_\epsilon^2\|\hat f_n-f_0\|_1
+2K_\epsilon\|\hat f_n'-f_0'\|_1
\xrightarrow{p}0.
\]
Lemma~\ref{lem:B-shrinkage} implies
\(\delta_n,\delta_0\in\mathcal D_{K_\epsilon}\).

Let \(J_n=J_{\hat G_n}\) and \(J_0=J_{G_0}\). 
Lemma~\ref{lem:B-continuity} gives, uniformly over
\(\delta\in\mathcal D_{K_\epsilon}\),
\[
|J_n(\delta)-J_0(\delta)|
\le(1-\epsilon)\Delta_n.
\]
The rules \(\delta_n\) and \(\delta_0\) minimize \(J_n\) and
\(J_0\), respectively.  Therefore
\begin{align*}
0
&\le J_0(\delta_n)-J_0(\delta_0)\\
&\le
|J_0(\delta_n)-J_n(\delta_n)|
+|J_n(\delta_0)-J_0(\delta_0)|\\
&\le2(1-\epsilon)\Delta_n.
\end{align*}
Lemma~\ref{lem:B-convexity} now yields
\begin{equation}
\label{eq:B-rule-convergence}
\int(\delta_n-\delta_0)^2f_0\le4\Delta_n.
\end{equation}

These bounds also control the maximum pointwise risk.  Write
\(\|h\|_0^2=\int h^2f_0\) and
\(\rho_j=r(G_0,\delta_j)\), \(j\in\{0,n\}\).  Factoring the difference
of squares and applying Cauchy--Schwarz gives
\begin{align*}
|\rho_n-\rho_0|
&\le \|\delta_n-\delta_0\|_0
\{\rho_n^{1/2}+\rho_0^{1/2}\}\\
&\le4\sqrt{\Delta_n/\epsilon},
\end{align*}
where the second line follows from
\eqref{eq:B-rule-convergence} and Lemma~\ref{lem:B-shrinkage}.  Since
\[
\epsilon\sup_{\theta\in\RR}R(\delta,\theta)
=J_0(\delta)-(1-\epsilon)r(G_0,\delta),
\]
the bound \(0\le J_0(\delta_n)-J_0(\delta_0)
\le2(1-\epsilon)\Delta_n\) proved above implies
\begin{align}
&\left|
\sup_{\theta\in\RR}R(\delta_n,\theta)
-\sup_{\theta\in\RR}R(\delta_0,\theta)
\right| \notag\\
&\qquad\le
\frac{1-\epsilon}{\epsilon}
\left\{2\Delta_n+4\sqrt{\Delta_n/\epsilon}\right\}
\xrightarrow{p}0.                       \label{eq:B-worst-risk}
\end{align}

The density \(f_0\) is continuous and strictly positive.  Moreover,
\(f_{G_0}^M\) is positive and smooth, so
\(\delta_0=x+(f_{G_0}^M)'/f_{G_0}^M\) is continuous.  Both
\(\delta_n\) and \(\delta_0\) are nondecreasing posterior means.
Lemma~\ref{lem:B-local-uniform} and
\eqref{eq:B-rule-convergence} therefore imply that, for every
\(0<M<\infty\),
\begin{equation}
\label{eq:B-local-rule-convergence}
\sup_{|x|\le M}|\delta_n(x)-\delta_0(x)|\xrightarrow{p}0
\end{equation}

We next transfer this convergence to the observations used to construct
\(\delta_n\).  Set
\[
A_n=\frac1n\sum_{i=1}^n
\{\delta_n(X_i)-\delta_0(X_i)\}^2.
\]
By Lemma~\ref{lem:B-shrinkage},
\(\sup_x|\delta_n(x)-\delta_0(x)|\le2K_\epsilon\).  Hence, for every
\(0<M<\infty\),
\begin{align}
A_n
&\le
\sup_{|x|\le M}|\delta_n(x)-\delta_0(x)|^2 \notag\\
&\quad+4K_\epsilon^2\frac1n\sum_{i=1}^n
\mathbf 1\{|X_i|>M\}.                    \label{eq:B-training-rule}
\end{align}
For fixed \(M\), the first term tends to zero by
\eqref{eq:B-local-rule-convergence}, while the empirical tail
frequency converges almost surely to
\[
\P(|X_1|>M)=\int_{|x|>M}f_0(x)\,dx.
\]
Taking first \(n\to\infty\) and then \(M\to\infty\) gives
\begin{equation}
\label{eq:B-training-L2}
A_n\xrightarrow{p}0.
\end{equation}

To pass from convergence of the rules to oracle-risk consistency, define the
realized average loss
\[
L_n(\delta)=\frac1n\sum_{i=1}^n
\{\delta(X_i)-\theta_i\}^2
\]
and put
\[
B_n=L_n(\delta_0)
=\frac1n\sum_{i=1}^n\{\delta_0(X_i)-\theta_i\}^2.
\]
The summands in \(B_n\) are i.i.d. and integrable, because
Lemma~\ref{lem:B-shrinkage} gives
\(r(G_0,\delta_0)\le1/\epsilon\).  The law of large numbers therefore
gives
\begin{equation}
\label{eq:B-oracle-loss}
B_n\longrightarrow r(G_0,\delta_0)
\qquad\text{almost surely}.
\end{equation}
Expanding the difference of squares and then applying
Cauchy--Schwarz yields
\begin{align}
|L_n(\delta_n)-L_n(\delta_0)|
&\le A_n+2\sqrt{A_nB_n}.                 \label{eq:B-loss-transfer}
\end{align}
Equations~\eqref{eq:B-training-L2}--\eqref{eq:B-loss-transfer} prove
that \(L_n(\delta_n)-L_n(\delta_0)\xrightarrow{p}0\).  Combining this
with \eqref{eq:B-oracle-loss} proves
\[
L_n(\hat\delta_n^M)\xrightarrow{p}r(G_0,\delta^M).
\]

It remains to pass to expectations.  Write
\(Z_i=X_i-\Theta_i\), so that the \(Z_i\)'s are i.i.d.
\(\NN(0,1)\).  Lemma~\ref{lem:B-shrinkage} gives
\[
0\le L_n(\delta_n)
\le \frac{2}{n}\sum_{i=1}^n Z_i^2+2K_\epsilon^2.
\]
Moreover,
\[
\E\left(\frac1n\sum_{i=1}^n Z_i^2\right)^2
=1+\frac{2}{n}\le3.
\]
The displayed upper bound therefore has a uniformly bounded second moment,
so \(\{L_n(\delta_n):n\ge1\}\) is uniformly integrable.  The
convergence in probability above therefore implies
\[
\mathcal R_n(\hat\delta_n^M)
=\E_{G_0}L_n(\delta_n)
\longrightarrow r(G_0,\delta^M).
\]
Since the fixed oracle rule satisfies
\[
\mathcal R_n(\delta^M)
=\E_{G_0}L_n(\delta_0)
=r(G_0,\delta^M)
\]
for every \(n\).  Consequently,
\[
\left|
\mathcal R_n(\hat\delta_n^M)
-\mathcal R_n(\delta^M)
\right|\longrightarrow0.
\]
Together with \eqref{eq:B-worst-risk}, this proves the theorem.
\end{proof}

\section{Worst Pointwise Risk of the Huber Rule}

In this appendix, we discuss the worst pointwise risk of the Huber rule for each of the
three DGPs considered in Section~\ref{sec: simu}. We then provide additional simulations
for the worst pointwise risks of the empirical Mallows and Huber rules and show that they
track their oracle counterparts closely for moderately large $n$.

For a given $\epsilon$, let $f_\epsilon^H$ denote the least favorable Huber density solving
\[
    \min_{F\in\mathcal F_\epsilon} I(F),
    \qquad
    \mathcal F_\epsilon
    =\{F=(1-\epsilon)(\Phi*G_0)+\epsilon W\},
\]
and write
\[
    g_\epsilon^H(x)=(\log f_\epsilon^H(x))',
    \qquad
    \delta_\epsilon^H(x)=x+g_\epsilon^H(x).
\]

We first state a simple observation used in all three cases. Suppose that a bounded,
absolutely continuous function $g$, with $\E_\theta|g'(X)|<\infty$ for every $\theta$,
satisfies
\begin{equation}\label{eq:C-risk-integrand-bound}
    g(x)^2+2g'(x)\leq k^2
\end{equation}
almost everywhere and $g(x)=-k$ for all sufficiently large $x$. For
$\delta(x)=x+g(x)$, Stein's identity gives
\[
    R(\delta,\theta)
    =1+\E_\theta\{g(X)^2+2g'(X)\}
    \leq 1+k^2.
\]
On the other hand, writing $X=\theta+U$ with $U\sim N(0,1)$, boundedness of $g$ and
dominated convergence give
\begin{align*}
    R(\delta,\theta)
    &=\E\{U+g(\theta+U)\}^2\\
    &\longrightarrow \E(U-k)^2=1+k^2
      \qquad\text{as }\theta\to\infty.
\end{align*}
Consequently, 
\begin{equation}\label{eq:C-sharp-bound}
    \sup_{\theta\in\RR}R(\delta,\theta)=1+k^2.
\end{equation}

\begin{proposition}
If $G_0=N(0,A)$, then
\[
g_\epsilon^H(x)=
\begin{cases}
    k_\epsilon, & x\leq-(A+1)k_\epsilon,\\
    -x/(A+1), & |x|<(A+1)k_\epsilon,\\
    -k_\epsilon, & x\geq(A+1)k_\epsilon,
\end{cases}
\]
and
\[
    \sup_{\theta\in\RR}R(\delta_\epsilon^H,\theta)
    =1+k_\epsilon^2
    =1+\left\{\sup_{x\in\RR}|g_\epsilon^H(x)|\right\}^2.
\]
\end{proposition}

\begin{proof}
The Huber rule is the limited translation rule of \cite{efronmorris}.
If $D=\sqrt{A+1}\,k_\epsilon$, then
$k_\epsilon^2=D^2/(A+1)$, where $D$ is determined by
\[
    (1-\epsilon)\{2\Phi(D)-1\}
    +\frac{2(1-\epsilon)\varphi(D)}{D}=1.
\]
In the two tails,
$(g_\epsilon^H)^2+2(g_\epsilon^H)'=k_\epsilon^2$. In the central region,
\[
    (g_\epsilon^H(x))^2+2(g_\epsilon^H)'(x)
    =\frac{x^2}{(A+1)^2}-\frac{2}{A+1}
    \leq k_\epsilon^2.
\]
Thus \eqref{eq:C-risk-integrand-bound} holds. Moreover,
\[
    \sup_x|g_\epsilon^H(x)|=k_\epsilon,
\]
and the score is constant in the right tail. Equation~\eqref{eq:C-sharp-bound}
therefore proves the result.
\end{proof}

\begin{proposition}
If $G_0=U[-B,B]$ for some $B>0$, then
\[
    \sup_{\theta\in\RR}R(\delta_\epsilon^H,\theta)
    =1+k_\epsilon^2
    =1+\left\{\sup_{x\in\RR}|g_\epsilon^H(x)|\right\}^2,
\]
where $k_\epsilon$ and $b$ satisfy
\[
    k_\epsilon=-(\log f_{G_0})'(b)
\]
and
\[
    1=(1-\epsilon)\left\{
        \int_{-b}^b f_{G_0}(x)\,dx
        +\frac{f_{G_0}(b)+f_{G_0}(-b)}{k_\epsilon}
      \right\}.
\]
\end{proposition}

\begin{proof}
Since both the uniform and Gaussian densities are log-concave, their convolution
$f_{G_0}$ is symmetric and log-concave. Hence
$g_0=(\log f_{G_0})'$ is odd and nonincreasing, with $g_0'\leq0$. The Huber score is
\[
g_\epsilon^H(x)=
\begin{cases}
    k_\epsilon, & x<-b,\\
    g_0(x), & |x|\leq b,\\
    -k_\epsilon, & x>b.
\end{cases}
\]
Symmetry and monotonicity give $|g_0(x)|\leq k_\epsilon$ on $[-b,b]$. Therefore,
on this interval,
\[
    g_0(x)^2+2g_0'(x)\leq k_\epsilon^2,
\]
while equality holds in both tails. Thus
\eqref{eq:C-risk-integrand-bound} holds,
$\sup_x|g_\epsilon^H(x)|=k_\epsilon$, and \eqref{eq:C-sharp-bound} proves the result.
\end{proof}

\begin{proposition}
Suppose
\[
    G_0=\frac12\delta_{-2}+\frac12\delta_2
\]
and $\epsilon$ is in the regime of Proposition~\ref{prop:Huber}. Then
\[
    \sup_{\theta\in\RR}R(\delta_\epsilon^H,\theta)
    =1+k_\epsilon^2
    =1+\left\{\sup_{x\in\RR}|g_\epsilon^H(x)|\right\}^2.
\]
\end{proposition}

\begin{proof}
By Proposition~\ref{prop:Huber}, the score has the form
\[
g_\epsilon^H(x)=
\begin{cases}
    k_\epsilon, & x\leq-b,\\
    g_0(x), & x\in[-b,-c]\cup[c,b],\\
    k_\epsilon\tanh(k_\epsilon x/2), & |x|\leq c,\\
    -k_\epsilon, & x\geq b,
\end{cases}
\]
where
\[
    g_0(x)=(\log f_{G_0}(x))'=-x+2\tanh(2x),
\]
and
\[
    g_0(c)=k_\epsilon\tanh(k_\epsilon c/2),
    \qquad
    g_0(b)=-k_\epsilon.
\]
Figure~\ref{fig:Huber-two-point-score} illustrates this score with $k_\epsilon=1$. 

It remains to verify \eqref{eq:C-risk-integrand-bound}. The equality is immediate in
the tails. If $h(x)=k_\epsilon\tanh(k_\epsilon x/2)$, then
\[
    h(x)^2+2h'(x)=k_\epsilon^2,
\]
so equality also holds on $[-c,c]$.

Now put $H(x)=g_0(x)^2+2g_0'(x)$ and $D(x)=g_0(x)-h(x)$. By the definition
of $c$ in this regime, $D(x)>0$ for $0<x<c$ and $D(c)=0$, so $D'(c)\leq0$.
Hence
\[
    H(c)-k_\epsilon^2=2D'(c)\leq0.
\]
At the other endpoint, $g_0(b)=-k_\epsilon$ and $g_0'(b)<0$, so
\[
    H(b)-k_\epsilon^2=2g_0'(b)<0.
\]
To control the interval between these endpoints, let $t=\tanh(2x)$. Direct
differentiation gives
\[
    H'(x)=F(t)
    :=(4t^2-3)\operatorname{arctanh}(t)-4t(5-4t^2).
\]
For $0<t\leq\sqrt{3}/2$, $F(t)<0$. For $t>\sqrt{3}/2$,
\[
    F'(t)
    =\frac{4t^2-3}{1-t^2}
      +8t\operatorname{arctanh}(t)-20+48t^2>0.
\]
Moreover, $F(\sqrt{3}/2)<0$ and $F(t)\to\infty$ as $t\uparrow1$. Thus $H$
first decreases and then increases on $(0,\infty)$, and its maximum on $[c,b]$
is attained at an endpoint. It follows that $H(x)\leq k_\epsilon^2$ on $[c,b]$;
symmetry gives the same conclusion on $[-b,-c]$.

Finally, $g_0''(x)<0$ for $x>0$, so $g_0$ is concave on $[c,b]$. Its minimum
is attained at an endpoint and is at least $-k_\epsilon$. If its maximum is attained
in the interior, then $g_0'=0$ there, and the inequality just proved gives
$g_0^2=H\leq k_\epsilon^2$. The central score has absolute value less than
$k_\epsilon$, while the tail score has absolute value exactly $k_\epsilon$.
Therefore $\sup_x|g_\epsilon^H(x)|=k_\epsilon$. Applying
\eqref{eq:C-sharp-bound} completes the proof.
\end{proof}

\begin{figure}
    \centering
    \includegraphics[scale=0.4]{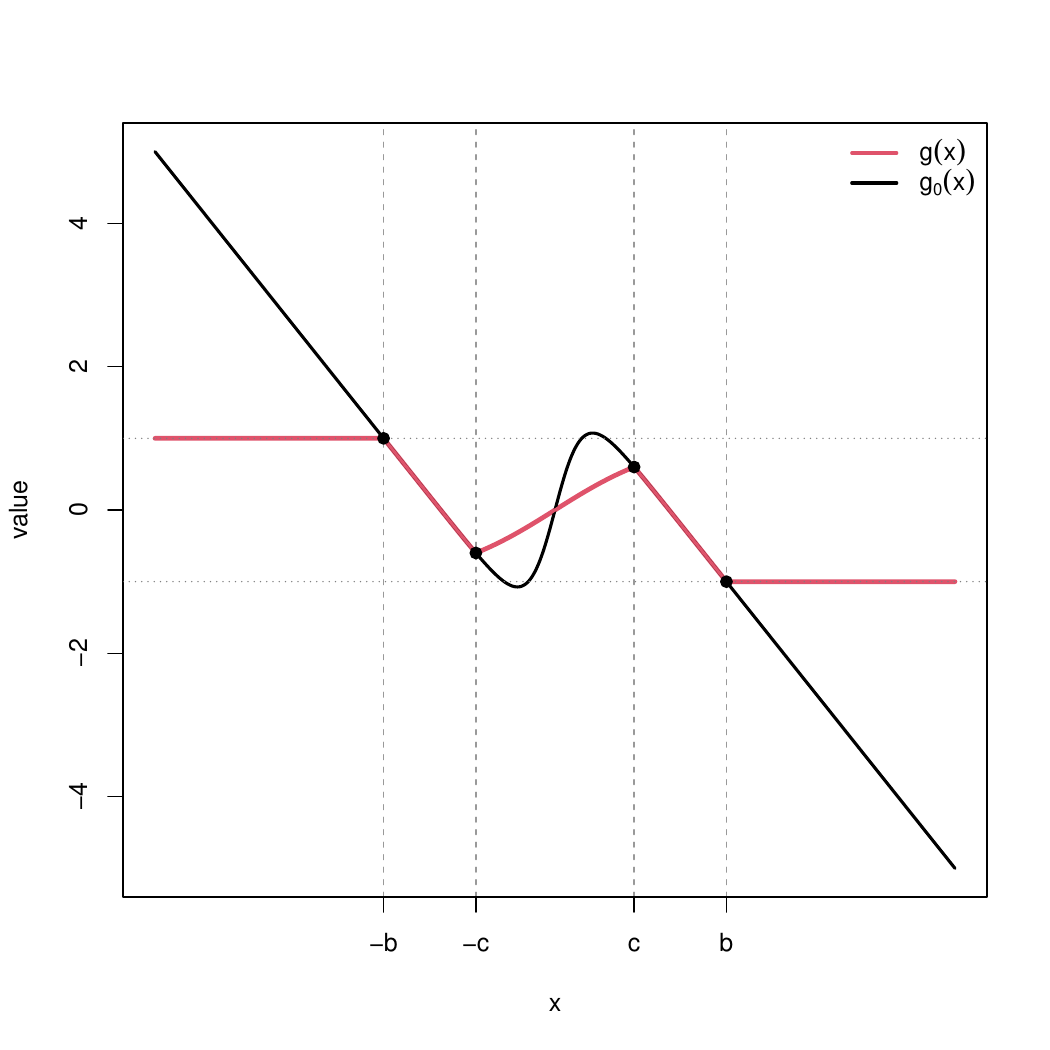}
    \caption{The Huber score for
    $G_0=\frac12\delta_{-2}+\frac12\delta_2$. The original score $g_0$ is unbounded,
    whereas the piecewise-defined Huber score has a central hyperbolic-tangent segment,
    agrees with $g_0$ on the two intermediate regions, and is constant in the tails.}
    \label{fig:Huber-two-point-score}
\end{figure}

We now present simulation results for the worst pointwise risks of the Mallows and Huber
rules and their empirical Bayes counterparts, obtained by replacing $G_0$ with the NPMLE
$\hat G$. For the Mallows rule, the positive-weight
support points of the least favorable contamination distribution $H^*$ maximize
$\theta\mapsto R(\delta^M,\theta)$ (see \cite{bickel1983minimizing} and \cite{kempthorne1987numerical}). In the numerical solution, we use the support point
$\theta_h$ having the largest computed weight, which avoids selecting spurious support
points with tiny numerical weights. 
For the Huber rule, the preceding results give the exact identity
\[
    \sup_{\theta\in\RR}R(\delta^H,\theta)
    =1+\left\{\sup_{x\in\RR}|g^H(x)|\right\}^2,
\]
and we approximate the score supremum on the computational grid. The results show that the
empirical worst-risk calculations approach their oracle counterparts as $n$ increases. They
also show that the worst pointwise risks decrease as $\epsilon$ increases.

\begin{table}[!tbp]
\begin{center}
\begin{tabular}{lrrrr}
\hline\hline
\multicolumn{1}{l}{}&\multicolumn{1}{c}{Huber-G0}&\multicolumn{1}{c}{Mallows-G0}&\multicolumn{1}{c}{Huber-Ghat}&\multicolumn{1}{c}{Mallows-Ghat}\tabularnewline
\hline
{\bfseries n =  100}&&&&\tabularnewline
~~$\epsilon = $ 0.05&$2.02$&$1.95$&$2.15$&$2.06$\tabularnewline
~~$\epsilon = $ 0.1&$1.68$&$1.64$&$1.73$&$1.69$\tabularnewline
~~$\epsilon = $ 0.2&$1.38$&$1.36$&$1.41$&$1.39$\tabularnewline
~~$\epsilon = $ 0.4&$1.16$&$1.15$&$1.17$&$1.16$\tabularnewline
\hline
{\bfseries n =  500}&&&&\tabularnewline
~~$\epsilon = $ 0.05&$2.02$&$1.95$&$2.14$&$2.06$\tabularnewline
~~$\epsilon = $ 0.1&$1.68$&$1.64$&$1.72$&$1.68$\tabularnewline
~~$\epsilon = $ 0.2&$1.38$&$1.36$&$1.40$&$1.39$\tabularnewline
~~$\epsilon = $ 0.4&$1.16$&$1.15$&$1.16$&$1.15$\tabularnewline
\hline
{\bfseries n =  1000}&&&&\tabularnewline
~~$\epsilon = $ 0.05&$2.02$&$1.95$&$2.06$&$1.99$\tabularnewline
~~$\epsilon = $ 0.1&$1.68$&$1.64$&$1.69$&$1.66$\tabularnewline
~~$\epsilon = $ 0.2&$1.38$&$1.36$&$1.39$&$1.37$\tabularnewline
~~$\epsilon = $ 0.4&$1.16$&$1.15$&$1.16$&$1.15$\tabularnewline
\hline
\end{tabular}
\caption{Worst case pointwise risk for various $\epsilon$, $G_0 = \NN (0,1)$ and various sample sizes\label{tab.sim1a_appendix}}\end{center}
\end{table}

\begin{table}[!tbp]
\begin{center}
\begin{tabular}{lrrrr}
\hline\hline
\multicolumn{1}{l}{}&\multicolumn{1}{c}{Huber-G0}&\multicolumn{1}{c}{Mallows-G0}&\multicolumn{1}{c}{Huber-Ghat}&\multicolumn{1}{c}{Mallows-Ghat}\tabularnewline
\hline
{\bfseries n =  100}&&&&\tabularnewline
~~$\epsilon = $ 0.05&$2.10$&$2.02$&$2.13$&$2.04$\tabularnewline
~~$\epsilon = $ 0.1&$1.66$&$1.64$&$1.71$&$1.67$\tabularnewline
~~$\epsilon = $ 0.2&$1.35$&$1.34$&$1.37$&$1.36$\tabularnewline
~~$\epsilon = $ 0.4&$1.13$&$1.12$&$1.15$&$1.13$\tabularnewline
\hline
{\bfseries n =  500}&&&&\tabularnewline
~~$\epsilon = $ 0.05&$2.10$&$2.02$&$2.19$&$2.11$\tabularnewline
~~$\epsilon = $ 0.1&$1.66$&$1.64$&$1.72$&$1.68$\tabularnewline
~~$\epsilon = $ 0.2&$1.35$&$1.34$&$1.37$&$1.36$\tabularnewline
~~$\epsilon = $ 0.4&$1.13$&$1.12$&$1.14$&$1.13$\tabularnewline
\hline
{\bfseries n =  1000}&&&&\tabularnewline
~~$\epsilon = $ 0.05&$2.10$&$2.02$&$2.12$&$2.05$\tabularnewline
~~$\epsilon = $ 0.1&$1.66$&$1.64$&$1.68$&$1.65$\tabularnewline
~~$\epsilon = $ 0.2&$1.35$&$1.34$&$1.36$&$1.35$\tabularnewline
~~$\epsilon = $ 0.4&$1.13$&$1.12$&$1.14$&$1.12$\tabularnewline
\hline
\end{tabular}
\caption{Worst case pointwise risk for various $\epsilon$, $G_0 = U[-2,2]$ and various sample sizes\label{tab.sim2a_appendix}}\end{center}
\end{table}

\begin{table}[!tbp]
\begin{center}
\begin{tabular}{lrrrr}
\hline\hline
\multicolumn{1}{l}{}&\multicolumn{1}{c}{Huber-G0}&\multicolumn{1}{c}{Mallows-G0}&\multicolumn{1}{c}{Huber-Ghat}&\multicolumn{1}{c}{Mallows-Ghat}\tabularnewline
\hline
{\bfseries n =  100}&&&&\tabularnewline
~~$\epsilon = $ 0.05&$3.02$&$2.41$&$2.74$&$2.25$\tabularnewline
~~$\epsilon = $ 0.1&$2.32$&$2.09$&$2.15$&$1.95$\tabularnewline
~~$\epsilon = $ 0.2&$1.69$&$1.67$&$1.61$&$1.59$\tabularnewline
~~$\epsilon = $ 0.4&$1.22$&$1.25$&$1.20$&$1.22$\tabularnewline
\hline
{\bfseries n =  500}&&&&\tabularnewline
~~$\epsilon = $ 0.05&$3.02$&$2.41$&$2.87$&$2.32$\tabularnewline
~~$\epsilon = $ 0.1&$2.32$&$2.09$&$2.23$&$2.02$\tabularnewline
~~$\epsilon = $ 0.2&$1.69$&$1.67$&$1.65$&$1.63$\tabularnewline
~~$\epsilon = $ 0.4&$1.22$&$1.25$&$1.21$&$1.23$\tabularnewline
\hline
{\bfseries n =  1000}&&&&\tabularnewline
~~$\epsilon = $ 0.05&$3.02$&$2.41$&$2.91$&$2.33$\tabularnewline
~~$\epsilon = $ 0.1&$2.32$&$2.09$&$2.25$&$2.04$\tabularnewline
~~$\epsilon = $ 0.2&$1.69$&$1.67$&$1.66$&$1.64$\tabularnewline
~~$\epsilon = $ 0.4&$1.22$&$1.25$&$1.21$&$1.23$\tabularnewline
\hline
\end{tabular}
\caption{Worst case pointwise risk for various $\epsilon$, 
	 $G_0 = 0.5 \delta_{-2} + 0.5 \delta_2$ and various sample sizes\label{tab.sim3a_appendix}}\end{center}
\end{table}

\begin{acks}
	The authors wish to express their appreciation to Pat Kline, Nikos Ignatiadis, Toru 
	Kitagawa, Peter Bickel and the editor, Lutz D\"umbgen, for comments on a previous draft.  
	Details on all computations for the 
	figures and tables are available from \url{https://rkoenker.github.io/www/roger/research/ebayes/ebayes.html}.
\end{acks}

\begin{supplement}
    \stitle{Appendix A:  Proof of Proposition 3.1}\label{app.A}
\end{supplement}
\begin{supplement}
    \stitle{Appendix B:  Proof of Theorem 3.2}\label{app.B}
\end{supplement}
\begin{supplement}
    \stitle{Appendix C:  Worst Pointwise Risk of the Huber Rule}\label{app.C}
\end{supplement}
\bibliographystyle{imsart-number}
\bibliography{REB}

\end{document}